\documentclass[11pt]{article}

\def\doctype{1}

\def\tsubmission{2}

\ifnum\doctype=\tsubmission
	\newcommand{\full}[1]{}
	\newcommand{\submit}[1]{#1}
\else
	\newcommand{\full}[1]{#1}
	\newcommand{\submit}[1]{}
\fi

\usepackage{graphicx}
\usepackage{epsfig}
\usepackage{amsfonts,amssymb,amstext,amsmath,amsthm}
\usepackage{xspace}
\usepackage{shadow,times}
\usepackage{algorithm,algorithmic}
\usepackage{fullpage}
\usepackage{enumerate}
\usepackage{fontenc,empheq}
\usepackage[utf8]{inputenc}
\usepackage{mathtools,cleveref}
\usepackage{paralist}
\usepackage{bm}
\usepackage{url}
\usepackage{prettyref}
\usepackage{boxedminipage}
\usepackage{wrapfig}
\usepackage{ifthen}
\usepackage{color}
\usepackage{xcolor}
\usepackage{framed}
\usepackage{thmtools}
\usepackage{thm-restate}

\usepackage{xstring}

\newcommand{\ignore}[1]{}
\allowdisplaybreaks

\newtheorem{theorem}{Theorem}
\newtheorem*{theorem*}{Informal Theorem} 

\newtheorem*{remark*}{Remark}
\newtheorem{lemma}[theorem]{Lemma}

\newtheorem*{claim*}{Claim}
\newtheorem{corollary}[theorem]{Corollary}

\renewenvironment{proof}{

\noindent{\bf Proof:}} {\hfill$\blacksquare$
}

\newenvironment{proofsketch}{

\noindent{\bf Proof Sketch:}} {\hfill$\blacksquare$
}

\newcommand{\initOneLiners}{%
    \setlength{\itemsep}{0pt}
    \setlength{\parsep }{0pt}
    \setlength{\topsep }{0pt}
}

\newenvironment{OneLiners}[1][\ensuremath{\bullet}]
    {\begin{list}
        {#1}
        {\initOneLiners}}
    {\end{list}}

\newenvironment{proofof}[1]{

\bigskip\noindent{\bf Proof of {#1}:}}
{\hfill$\blacksquare$
}

\newcommand{\sse}{\subseteq}
\def\etal{{\em et al.}\xspace}

\DeclareMathOperator*{\E}{\mathbb{E}}

\def\Expectation#1{\E\left[#1\right]}

\def\opt{\ensuremath{\mathsf{Opt}}\xspace}
\def\junc{\ensuremath{\mathsf{junc}}\xspace}
\def\ssalg{\ensuremath{\mathsf{SingleSinkAlg}}\xspace}

\def\upp{\ensuremath{\mathrm{up}}\xspace}
\def\downn{\ensuremath{\mathrm{down}}\xspace}

\def\siti{\ensuremath{(s_i,t_i)}\xspace}

\newcommand{\INDSTATE}[1][1]{\STATE\hspace{#1\algorithmicindent}}

\newcommand{\comment}[1]{\textit {\em \color{blue} \footnotesize[#1]}\marginpar{\tiny\textsc{\color{blue} To Do!}}}

\newcommand{\eat}[1]{}

\newcommand{\cI}{{\cal I}}
\newcommand{\cJ}{{\cal J}}

\newcommand{\cP}{\mathcal{P}}

\newcommand{\cS}{\mathcal{S}}
\newcommand{\cT}{{\cal T}}

\newcommand{\cX}{{\cal X}}

\newcommand{\R}{\mathbb R}

\newcommand{\eps}{\varepsilon}

\newcommand{\polylog}{\mathrm{polylog}}

\newcommand{\ceil}[1]{\lceil#1\rceil}

\def\load{{\sf load}}
\def\Obj{{\sf Obj}}

\title{Online Buy-at-Bulk Network Design}

\author{Deeparnab Chakrabarty\thanks{Microsoft Research, 9 Lavelle Road, Bangalore, India. Email: {\tt dechakr@microsoft.com}.}
\and Alina Ene\thanks{Department of Computer Science and DIMAP, University of Warwick, Coventry, UK. Email: {\tt A.Ene@warwick.ac.uk}.}
\and Ravishankar Krishnaswamy\thanks{Microsoft Research, 9 Lavelle Road, Bangalore, India. Email: {\tt rakri@microsoft.com}.}
\and Debmalya Panigrahi\thanks{Department of Computer Science, Duke University, Durham, NC, USA. Email: {\tt debmalya@cs.duke.edu}.}}
\date{}

\begin{document}

\maketitle

\thispagestyle{empty}
\begin{abstract}
We present the first non-trivial online algorithms for the non-uniform, multicommodity
buy-at-bulk (MC-BB) network design problem. Our competitive ratios
qualitatively match the best known approximation factors for the
corresponding offline problems. In particular, we show
\begin{itemize}
	\item A polynomial time  online algorithm with a poly-logarithmic competitive
	ratio for the MC-BB problem in undirected edge-weighted graphs.
	\item A quasi-polynomial time online algorithm with a poly-logarithmic
	competitive ratio for the MC-BB problem in undirected node-weighted graphs.
	\item For any fixed $\epsilon > 0$, a polynomial time online algorithm with
	a competitive ratio of $\tilde{O}\big(k^{\frac{1}{2}+\epsilon})$ (where $k$ is
	the number of demands, and $\tilde{O}(.)$ hides polylog factors) for MC-BB in directed graphs.
    \item Algorithms with matching competitive ratios for the prize-collecting variants of all the above problems.
\end{itemize}    
Prior to our work, a logarithmic competitive ratio was
known for undirected, edge-weighted graphs only for the special case of {\em uniform} costs (Awerbuch and Azar,
FOCS 1997), and a polylogarithmic competitive ratio was known for the edge-weighted {\em single-sink} problem (Meyerson, SPAA 2004).
To the best of our knowledge, no previous online algorithm was known,
even for uniform costs, in the node-weighted and directed settings. \smallskip

Our main engine for the results above is an {\em online reduction theorem} of MC-BB problems
to their single-sink (SS-BB) counterparts. We use the concept of {\em junction-tree solutions} (Chekuri~\etal, FOCS 2006)
that play an important role in solving the {\em offline} versions of the problem via a greedy subroutine -- an inherently
offline procedure. Our main technical contribution is in designing an online algorithm using only the 
{\em existence} of good junction-trees to reduce an MC-BB instance to multiple SS-BB sub-instances.
Along the way, we also give the first non-trivial online node-weighted/directed single-sink buy-at-bulk algorithms.
In addition to the new results, our generic reduction also yields new proofs of recent results for
the online node-weighted Steiner forest and online group Steiner forest problems.

\end{abstract}

\newpage
\section{Introduction}
\label{sec:introduction}

In a typical network design problem, one has to find a minimum cost (sub) network satisfying
various connectivity and routing requirements. These are fundamental problems
in combinatorial optimization, operations research,  and computer science.
To model economies of scale in network design, Salman~\etal~\cite{SalmanCRS01} proposed the {\em buy-at-bulk} framework, which
has been studied extensively over the last two decades
(e.g.,~\cite{AwerbuchA97, GuhaMM09, Talwar02, GuptaKR03, MeyersonMP08, Meyerson04, ChekuriHKS10, ChekuriHKS07}).
In this framework, each network element is associated with a sub-additive function representing the cost  for a given utilization.
Given a set of connectivity demands comprising $k$ source-sink pairs, the goal is to route
integral flows from the sources to the corresponding sinks concurrently to minimize the total cost  of the routing.


An important application of the problem is capacity planning in telecommunication networks
or in the Internet.
As observed by Awerbuch and Azar~\cite{AwerbuchA97}, this application is inherently
``online" in that terminal-pairs arrive over time and need to be served without knowledge of future pairs.
The authors of~\cite{AwerbuchA97} give a logarithmic-competitive 
online algorithm for the  {\em uniform} case where every edge is
associated with the same cost function.
However, uniformity is not always a feasible assumption, especially in
heterogeneous, dynamic networks like the Internet. Indeed recent research (e.g.,~\cite{MeyersonMP08, Meyerson04, ChekuriHKS10}) has
focused on the non-uniform setting with a different
sub-additive function for every network element.
In this non-uniform setting, Meyerson~\cite{Meyerson04} gives a polylogarithmic-competitive algorithm for the special case when all terminal-pairs share the same sink.
To the best of our knowledge, no non-trivial online algorithm is known for the general {\em multicommodity} setting, which is the focus of our paper.

We consider, in increasing order of generality,
{\em undirected edge-weighted} graphs, {\em undirected node-weighted} graphs,
and {\em directed edge-weighted} graphs.\footnote{In undirected graphs, node
costs can simulate edge costs; in
directed graphs they are equivalent.} It is also convenient to classify the problems that we study into
the {\em single-sink} 
version where all the terminal-pairs share a common sink,
and the general {\em multicommodity} 
version where the sinks in the terminal-pairs may be distinct.
For notational convenience, we use the following shorthand forms for our problems:
X-Y-BB where X = SS or MC (single-sink and multicommodity,  respectively)
and Y = E or N or D (undirected edge-weighted, undirected node-weighted,
and the general directed case, respectively).

\subsection{Our Contributions}\label{sec:contrib}
We obtain the following new results (unless otherwise noted, our algorithms
run in polynomial time):
\begin{itemize}
	\item A poly-logarithmic competitive online algorithm for the MC-E-BB problem.
	\item A poly-logarithmic competitive online algorithm for MC-N-BB and
	SS-N-BB that runs in quasi-polynomial time.
	\item An $\tilde{O}(k^{\frac{1}{2}+\eps})$-competitive online algorithm for MC-D-BB for any constant $\eps> 0$ with running time $n^{O(1/\eps)}$,
	where $\tilde{O}(.)$ hides polylogarithmic factors. For SS-D-BB, the
	ratio improves to $\tilde{O}(k^{\epsilon})$, translating to a
	polylogarithmic	competitive ratio in quasi-polynomial time.
	\item Online algorithms for prize-collecting versions of all the above problems with the same competitive ratio.
%
\end{itemize}

\noindent
Up to exponents in the logarithm, our online algorithms match the best known {\em offline} approximation algorithms (Chekuri~\etal~\cite{ChekuriHKS10} for MC-E/N-BB and Antonakopoulos~\cite{Antonakopoulos10} for MC-D-BB); for MC-N-BB, however, a {\em polynomial time}, polylogarithmic approximation is known~\cite{ChekuriHKS07}, whereas our algorithm runs in quasi-polynomial time. Furthermore, a logarithmic lower bound, even for SS-E-BB,
follows from the lower bound for the online Steiner tree problem~\cite{ImaseW91}, and a polylogarithmic lower bound for online SS-N-BB follows from a matching one for set cover~\cite{AlonAABN09}.

From a technical perspective, we derive all the multicommodity
results using a generic online reduction theorem that reduces a multicommodity instance to several single-sink instances, for which we either use existing online algorithms or give new online algorithms. Informally, one can view this as the ``online analog'' of the junction-tree approach pioneered by Chekuri \etal~\cite{ChekuriHKS10} for offline multicommodity network design. We discuss this approach in the next subsection.

\subsection{An Online Reduction to Single Sink Instances}
Multicommodity network design problems, both online and offline,  are typically more challenging
than their single-sink counterparts, and have historically\footnote{For instance, compare~\cite{MeyersonMP08} and~\cite{CharikarK05} for the SS-E-BB and MC-E-BB problem, and
compare Naor~\etal~\cite{NaorPS11} and Hajiaghayi~\etal~\cite{HajiaghayiLP13} for the online node-weighted Steiner tree and Steiner forest problem.}
required new ideas every time depending on the specific problem at hand.
%
%
The situation is no different for buy-at-bulk, both for
uniform and non-uniform costs.
In the offline buy-at-bulk setting, this
shortcoming is addressed by Chekuri~\etal~\cite{ChekuriHKS10}
(expanded to other problems by \cite{ChekuriHKS07,ChekuriEGS11,Antonakopoulos10}),
who introduce a generic combinatorial framework for
mapping a single instance of a multicommodity problem to multiple instances of the
corresponding single-sink problem. At the heart of this scheme is the following
observation that holds for many multicommodity problems such as (edge/node) Steiner forest, directed
Steiner network, buy-at-bulk, and set connectivity:
 there exists a near-optimal\footnote{We call the quality of such a solution the junction-tree approximation factor; e.g., it is $O(\log n)$ for MC-E-BB and MC-N-BB~\cite{ChekuriHKS10}} {\em junction-tree} solution for the multicommodity problem
that decomposes into solutions to multiple single-sink problems where each single-sink problem connects some subset of the original terminal-pairs to a particular root.


The problem now reduces to finding good junction-trees to cover all the terminal-pairs. 
 The offline
 techniques~\cite{ChekuriHKS10,ChekuriEGS11,Antonakopoulos10} tackle
 this using a greedy algorithm for finding the single-sink solutions;
 more precisely, in each step they find the best
 density (cost per terminal-pair) solution that routes a subset of
 terminal-pairs via a single sink. A set cover style analysis then
 bounds the loss for repeating this procedure until all terminal-pairs are
 covered. However, as the reader may have already noticed, the greedy
 optimization approach is inherently offline, as finding the
 best-density solution requires knowledge of all terminal-pairs upfront. Our
 main technical contribution in this work is \emph{an online
 version of the junction-tree framework}. Indeed, we show how to
 reduce any multicommodity buy-at-bulk instance to a collection of
 single-sink instances online.

\smallskip\noindent
{\bf (Informal Theorem)} If the junction-tree approximation factor of the MC-BB problem is $\alpha$, the integrality gap of a natural LP relaxation of the SS-BB problem is $\beta$, and there is a $\gamma$-competitive online algorithm for the SS-BB problem, then there is an $O(\alpha\beta\gamma \cdot \polylog(n))$-competitive algorithm for the MC-BB problem.

To prove the above theorem, we first write a composite-LP, which has (a) an {\em outer-LP} comprising assignment variables that fractionally assign terminal-pairs to roots, and (b) many {\em inner-LP}s which correspond to the natural LP relaxations for the SS-BB problem for each root and the terminal-pairs fractionally assigned to it by the outer-LP. We then apply the framework of online primal-dual algorithms (see~\cite{BuchbinderN09} for instance) to solve the composite-LP online. However, there are two main challenges we need to surmount.

$\bullet$ First, the existing framework has been mostly applied to
purely covering/packing LPs\footnote{Our current understanding of
mixed packing-covering is limited~\cite{AzarBFP13} and does not
capture the problem we want to solve.} and our inner-LPs have both
kinds of constraints, and moreover, there is an outer-LP
encapsulating them. We show nonetheless that it can be extended to solving our LP fractionally up to polylogarithmic factors. Indeed, we use the specific flow-structure of the inner-LP, and each step of our algorithm solves many (auxiliary) min-cost max-flow problems.

$\bullet$ The second difficulty is in {\em rounding} this fractional solution online. This is a hard problem, and currently we do not know how to do so even for basic network design problems such as the
Steiner tree problem. To circumvent this, we show that it suffices to only
{\em partially round} the LP. More precisely, we round the LP
solution so that only the outer-LP (assignment variables) become
integral, and the inner-LPs remain fractional. This gives us an
\emph{integral assignment} of the terminal-pairs to different
single-sink instances, with bounded total fractional cost. Now, from
the bounded integrality gap of the inner-LPs, we know that
\emph{there exist good single-sink solutions} for our assignment of
terminal-pairs to roots, even though we cannot find them
online\footnote{The difficulty comes from the fact that the online solution we maintain must be monotonic, i.e., the decisions are irrevocable.}. Using this knowledge, our final step is to run online single-sink algorithms for each root, and send the terminal-pairs to the root as determined by the outer-LP assignment.
Figure~\ref{alg-framework} summarizes our overall approach.

\begin{figure}[h!]
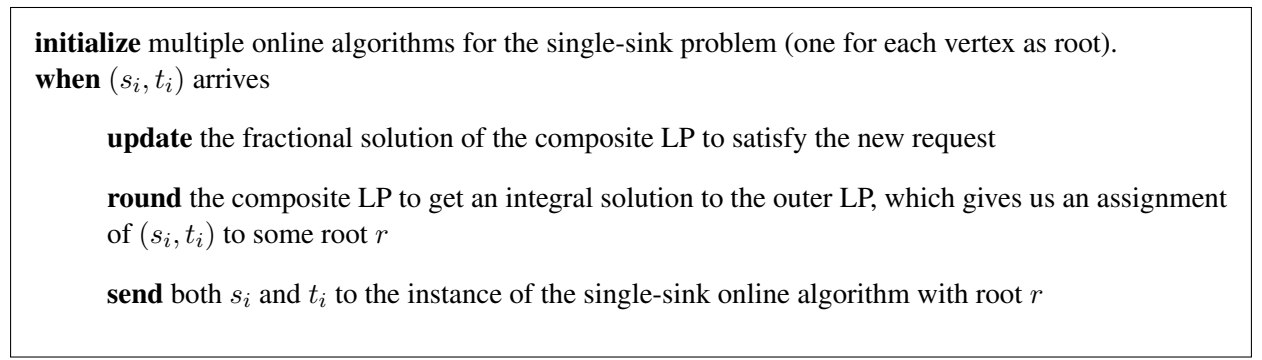

\begin{framed}
{\bf initialize} multiple online algorithms for the single-sink problem (one for each vertex as root).\\
{\bf when} \siti arrives
\begin{itemize}
\item[]  {\bf update} the fractional solution of the composite LP to satisfy the new request
\item[]  {\bf round} the composite LP to get an integral solution to the outer LP, which gives us an assignment of \siti to some root $r$
\item[]  {\bf send} both $s_i$ and $t_i$ to the instance of the single-sink online algorithm with root $r$
\end{itemize}
\end{framed}
\caption{{\bf Online Framework for Multicommodity Network Design Problems}}
\label{alg-framework}
\end{figure}

The results mentioned in~\Cref{sec:contrib} follow
by  bounding $\alpha,\beta,\gamma$ for the corresponding problems. For MC-E-BB, all of these are known to be $O(\polylog(n))$ (\cite{ChekuriHKS10,MeyersonMP08,Meyerson04} respectively).
For MC-N-BB, it is known both $\alpha,\beta$ are bounded by $O(\polylog(n))$~\cite{ChekuriHKS07}, and we bound $\gamma$\full{ in~\Cref{sec:node-weighted}} by giving the first online algorithms for SS-N-BB.
For MC-D-BB, we need some additional work. In this case we cannot directly bound $\beta$, since the integrality gap of the natural LP relaxation is {\em not known} to be bounded. Nevertheless,\full{ in~\Cref{sec:directed},} we show that it suffices to work only with more structured instances for which we can bound the integrality gap. 

Finally, we illustrate the generality of our reduction theorem by
noting that, when combined with existing bounds on $\alpha, \beta$, and $\gamma$, it immediately implies
(up to polylogarithmic factors) some recent results in online network design, such as
online node-weighted Steiner forest~\cite{HajiaghayiLP13}, and online edge-weighted group Steiner forest~\cite{NaorPS11} -- two problems for which specialized techniques were needed, even though their single-sink counterparts were known earlier.

\subsection{Related Work}

Buy-at-bulk network design problems have received considerable
attention over the last two decades, both in the offline and
online settings. For the uniform cost model, Awerbuch and Azar~\cite{AwerbuchA97} give an $O(\log n)$-approximation for MC-E-BB,
while $O(1)$-approximations are known~\cite{GuhaMM09, Talwar02, GuptaKR03} for SS-E-BB.
We also note that $O(1)$-approximations have been obtained in special cases for the multicommodity problem, such as in the {\em rent-or-buy} setting~\cite{GuptaKPR03}.
Meyerson~\etal~\cite{MeyersonMP08} give an $O(\log k)$ approximation for the general SS-E-BB, and the first non-trivial algorithm for MC-E-BB
is an $\exp(O(\sqrt{\log n\log\log n}))$-approximation due to Charikar and Karagiozova~\cite{CharikarK05}.
This was improved to a poly-logarithmic factor by Chekuri \etal \cite{ChekuriHKS10} who also solve MC-N-BB~\cite{ChekuriHKS07} with similar guarantees.
For directed graphs, our knowledge is much sparser. Even for special cases like directed Steiner tree and forest, the best polytime approximation factors known are $O(k^\eps)$~\cite{Zelikovsky97,CharikarCCDGGL99} and  $\min(O(\sqrt{k},n^{2/3}))$~\cite{ChekuriEGS11,FeldmanKN12,BermanBMRY13} respectively, and these ideas were extended to MC-D-BB by Antonakopoulos~\cite{Antonakopoulos10}.
On the hardness side,  Andrews~\cite{Andrews04} shows that even the MC-E-BB problem is $\Omega\big(\log^{1/2 - \eps} n)$-hard,
while MC-D-BB (in fact directed Steiner forest) is known to be label-cover hard~\cite{DodisK99}.

The online Steiner tree problem (a special case of online SS-E-BB) was first studied by Imase and Waxman~\cite{ImaseW91} who give an $O(\log k)$-competitive algorithm.
Berman and Coulston~\cite{BermanC97} give an $O(\log k)$-competitive algorithm for online Steiner forest, and both these results are tight, i.e., there is an $\Omega(\log k)$ lower bound.
As mentioned earlier, Awerbuch and Azar's algorithm~\cite{AwerbuchA97} can be seen as an $O(\log n)$-competitive online algorithm for the {\em uniform-cost} MC-E-BB. For non-uniform buy-at-bulk, 
the only online algorithm that we are aware of is Meyerson's~\cite{Meyerson04} polylog-competitive algorithm for the single-sink problem.
For online node-weighted network design, developments are much more recent.
 A polylogarithmic approximation for the node-weighted Steiner tree problem was first given by Naor~\etal~\cite{NaorPS11}
and later extended to the online node-weighted Steiner forest
problem~\cite{HajiaghayiLP13} and prize-collecting
versions~\cite{HajiaghayiLP14}. These algorithms, like ours
in this paper, utilize the online adaptation of the primal-dual
and LP rounding schemas pioneered by the work of Alon~\etal~\cite{AlonAABN09}
for the online set cover problem (see also
\cite{AlonAABN06} for its adaptation to network design problems).
We also note that in the node-weighted setting,
the online lower bound can be strengthened to $\Omega(\log n \log k)$
	using online set cover lower bounds~\cite{AlonAABN09,Korman05}.

\section{Preliminaries and Results}
\label{sec:preliminaries}

We now formally state the problem, set up notation
that we use throughout the paper, and state our main theorems.

\subsection{Our Problems}

\noindent
{\bf Buy-at-bulk Network Design.}
In the most general setting of the MC-D-BB problem, an instance $\cI = (G,\cX)$ consists of a directed graph $G = (V,E)$ and a
collection $\cX$ of \emph{terminal-pairs} $(s_i, t_i)\in V\times V$; each such
$s_i$ and $t_i$ is called a terminal. Each \siti pair also
has a positive integer demand $d_i$, which we assume
to be $1$ for clarity in presentation.\footnote{We can handle non-uniform demands by incurring an additional $O(\log D)$ factor in the competitive ratio and the running time (where $D$ is the maximum demand) by having $O(\log D)$ ``unit-demand'' instances, where the $i^{th}$ instance deals with demands between $2^{i-1}$ and $2^{i}$.}
Additionally, each edge $e \in E$ is associated with a {\em monotone, sub-additive}\footnote{That
is, $f_e(x) \geq f_e(y)$ whenever $x\geq y$, and $f_e(x+y) \leq
f_e(x) + f_e(y)$} cost function $f_e:\R_{\geq 0}\to\R_{\geq 0}$. A feasible solution to the problem is a collection of paths $\{P_1,
\ldots, P_k \}$ where $P_i$ is a directed path from $s_i$ to
$t_i$ carrying load $d_i$.
Given a solution $\{P_1, \ldots, P_k\}$, we let $\load(e) = \sum_{i : e \in P_i} d_i$ denote the
total load on edge $e$. The goal is to find a feasible solution
minimizing the objective $\Obj_{\mathsf{BB}}  := \sum_{e\in E} f_e(\load(e))$.
In the \textbf{online} problem, the offline input consists of the
graph $G$ and the cost functions $f_e$. The pairs \siti arrive online
in an unknown, possibly adversarial, order.  When a
pair \siti arrives, the algorithm must select the path $P_i$ that
connects them, and this decision is irrevocable. 

\medskip
\noindent
{\bf Reduction to the Two-metric Problem.}
Following previous work, throughout this paper we consider an equivalent problem (up to constant factors)
known as {\bf two-metric network design}.  In this
problem, instead of functions $f_e(.)$ on the edges, we are given two
parameters $c_e$ and $\ell_e$ on each edge. One can think of $c_e$ as
a fixed buying cost, or just {\em cost}, of edge $e$, and $\ell_e$ as a per-unit
flow cost, or {\em length}, of edge $e$.  The feasible solution space is the same as
for the buy-at-bulk problem, and the goal is to minimize the
objective $\Obj_{\mathsf{2M}} := \sum_{e \in \bigcup_i P_i} c_e +
	\sum_{i} \sum_{e \in P_i} \ell_e$.
The following lemma is well known (see e.g.,
\cite{ChekuriHKS10}).

\begin{lemma}
	Given an instance of the buy-at-bulk problem, for any $\eps>0$ one can find an
	instance of the two-metric network design problem such that, for any
	feasible solution, $\Obj_{\mathsf{2M}} \leq \Obj_{\mathsf{BB}} \leq
	(2+\eps)\Obj_{\mathsf{2M}}$.
\end{lemma}

\begin{remark*}
	In light of the above lemma, henceforth we abuse notation and let the buy-at-bulk problem mean the two-metric network design problem.
\end{remark*}

\subsection{Our Tools}

\noindent
{\bf Junction-tree solutions.}
Given an instance $\cI = (G, \cX)$ of the buy-at-bulk problem, we
consider \emph{junction-tree\footnote{The word tree is misleading
since the final solution need not be a tree in directed graphs. Nevertheless, 
we continue using this term for
historical reasons. Junction trees were originally proposed for undirected 
graphs, where the solution is indeed a tree.} solutions}, a
specific kind of solution to the problem introduced by~\cite{ChekuriHKS10}.  In such solutions, the collection of pairs are
partitioned into groups and each group is indexed by a {\em root
vertex} $r\in V$. For all terminal pairs \siti in a group indexed by
$r$, the path $P_i$ from $s_i$ to $t_i$ contains the root vertex $r$ (see
Figure~\ref{fig:junction}).

\begin{figure}[h]
	\centering
	\includegraphics{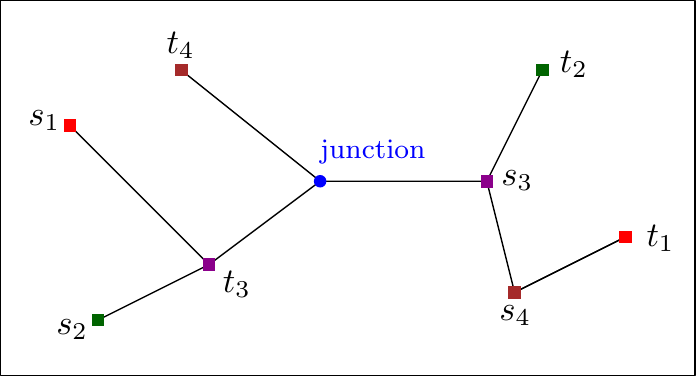}
	\caption{A group of terminal pairs routed via a junction-vertex in an undirected graph.}
	\label{fig:junction}
\end{figure}

Formally, consider an instance $\cI = (G,\cX)$ of the buy-at-bulk problem, and
let $\opt$ denote the objective value of the optimum solution.  Given
a partition $\Pi := (\pi_{r_1},\ldots,\pi_{r_q})$ of terminal pairs
indexed by $q$ different root vertices, a junction-tree solution is one that uses {\em
single-sink} solutions to connect the original terminal-pairs. Indeed, for
each part $\pi_r$ indexed by root $r$, consider the optimal solutions to the
single-sink problem on graph $G$ with demands $\{(s_i,r) : \siti \in \pi_r\}$ and the
single-source problem\footnote{The single-source problem in
a directed graph is identical to the single-sink problem with all the edges
reversed in direction. For undirected graphs, both problems are on the
same graph.} with pairs $\{(r,t_i) : \siti \in \pi_r\}$. 
Let $\opt_r(\pi_r)$ denote the sum of the objectives of the
optimal solutions to the single-sink and single-source problems, and let
$\opt(\Pi) := \sum_{r\in V} \opt_r(\pi_r)$.  Let $\opt_\junc$ denote
the minimum $\opt(\Pi)$ over all partitions. We call this solution the {\em
optimum junction-tree} solution for this instance.\footnote{Note that
copies of the same edge appearing in multiple single-sink solutions
are treated as distinct edges in the junction-tree solution. Hence,
decomposing the optimal multicommodity solution into its constituent
paths does not yield $\opt_\junc = \opt$.} Clearly,
$\opt_\junc \geq \opt$. The {\bf junction-tree approximation factor}
of $\cI$ is defined to be the ratio $\opt_\junc/\opt$.

\medskip
\noindent
{\bf LP Relaxation.}
We now describe a natural flow-based LP relaxation for the {\em
single-sink} buy-at-bulk problem for an instance $\cI = (G,\cT)$
where $\cT$ is a set of terminals that need to be connected to the
root $r$.

\begin{subequations}
\begin{alignat}{2}
		\text{minimize }   & \quad \sum_{e \in E} c_e x_{e} + \sum_{i}
		\sum_{e \in E} \ell_e f_i(e) \tag{SS-BaB LP} \label{sslablp}\\
		\text{s.t }    & \quad\{f_i(e) \colon e \in E(G)\} \; \text{ defines a flow from
		} s_i \text{ to } r \text{ of value } 1 &\quad \forall s_i \in
		\cT \notag \\
                       &\quad f_i(e) \leq x_e &\quad \forall e\in E \notag \\
                       & \quad x_e\geq 0,\ f_i(e) \geq 0 &\quad
											 \forall e \in E \notag
\end{alignat}
\end{subequations}

\noindent
Recall that the \emph{integrality gap} of \eqref{sslablp} on the
instance $\cI = (G, \cT)$ is defined to be the ratio of $\opt$ to the
optimal value of the LP \eqref{sslablp}. Also, we define the integrality gap
for the graph $G$ to be the worst case integrality gap (over all requests $\cT$ on
graph $G$) of the corresponding instance $\cI = (G,
\cT)$.

\subsection{Our Results}

\noindent
{\bf Main Technical Theorem and its Applications.}
Now we are ready to state our main theorem; the proof is in Section~\ref{sec:proofmain}. We say that an online
algorithm is $\gamma$-competitive for a graph $G$ if, for any
sequence of requests $\cX$,
the online algorithm for buy-at-bulk returns a solution within
a $\gamma$-factor of $\opt(\cI)$, where $\cI = (G,\cX)$.

\begin{theorem} [{\bf Reduction to Single-Sink Online Algorithms}]
\label{thm:main}
Fix an instance $\cI = (G,\cX)$ of the MC-BB problem.
Suppose the following three conditions hold.
\begin{OneLiners}
	\item [(i)]  The {junction-tree approximation factor} of $~\cI$ is
	at most $\alpha$.
	\item [(ii)] The integrality gap of \eqref{sslablp} on any
	{\em single-sink} instance on graph $G$ is at most
	$\beta$.
	\item [(iii)] There is a $\gamma$-competitive online SS-BB algorithm for any instance on graph $G$ that runs in time $T$.
\end{OneLiners}
Then there is an online algorithm for $\cI$ running in time
$\textrm{poly}(n,T)$ whose competitive ratio is $O(\alpha \beta
\gamma\cdot \polylog(n))$.
\end{theorem}

\noindent
Using this theorem, we can immediately obtain the following new results mentioned in the introduction.

\begin{theorem}[Undirected Edge-weighted Buy-at-Bulk]
\label{thm:undirected-edge-bab}
	There is a $\polylog(n)$-competitive, polynomial time
	randomized online algorithm for the MC-E-BB problem.
\end{theorem}
\begin{proof}
	The theorem follows directly by combining Theorem~\ref{thm:main}
	with the following results from previous work. Chekuri~\etal~\cite{ChekuriHKS10} prove that the junction-tree
	approximation factor for the undirected edge-weighted buy-at-bulk
	problem is $O(\log k)$. Chekuri~\etal~\cite{ChekuriKN01} prove
	that the integrality gap of \eqref{sslablp} in undirected
	edge-weighted graphs is $O(\log k)$. Meyerson~\cite{Meyerson04}
	gives a randomized polynomial time online algorithm for the
	single-sink buy-at-bulk problem with competitive ratio
	$O(\log^4{n})$.
\end{proof}


\begin{theorem}[Undirected Node-weighted Buy-at-Bulk]
\label{thm:undirected-node-bab}
		For any constant $\eps > 0$, there is an $O(k^{\eps}\polylog(n))$-competitive,
		randomized online algorithm for MC-N-BB with running time $n^{O(1/\eps)}$. As a corollary, this yields a $\polylog(n)$-competitive,
		quasi-polynomial time algorithm for this problem.
\end{theorem}

\begin{theorem}[Directed Buy-at-Bulk]
\label{thm:gen-bab}
	For any constant $\eps>0$, there is an $O\big(k^{1/2 +
	\eps}\polylog(n)) \big)$-competitive, polynomial time online
	algorithm for the MC-D-BB.
\end{theorem}

\submit{
\begin{theorem} \label{thm:prizes}
(Prize Collecting Buy-at-Bulk)
For each of the above problems, there is an online algorithm with matching running time and competitive ratio for the corresponding prize-collecting\footnote{In the prize-collecting problem, every terminal-pair also comes with a penalty $q_i$, and the algorithm can opt to not satisfy the request by incurring this value in the objective.} version.
\end{theorem}
}
\full{
\noindent
We again use~\Cref{thm:main} to prove the above theorems. However, unlike for MC-E-BB, we are not aware of any
online algorithms for the SS-N-BB and SS-D-BB problems. We therefore first give online algorithms for
these problems, and then use~\Cref{thm:main}; the details appear in Sections~\ref{sec:node-weighted} and
\ref{sec:directed}.
}
\submit{
We again use~\Cref{thm:main} to prove the above theorems. However, unlike for MC-E-BB, we are not aware of any
online algorithms for the SS-N-BB and SS-D-BB problems. We therefore first give online algorithms using fairly standard techniques for
these problems, and then use~\Cref{thm:main}. Finally, we can almost directly use~\Cref{thm:main} to also obtain matching results for prize-collecting versions of the above problems.
Due to space constraints, we defer the proofs of the above to the full version of the paper.
}

\full{\smallskip \noindent Finally, we can almost directly
use~\Cref{thm:main} to also obtain matching results for
prize-collecting versions of the above problems. Recall that in a
prize-collecting problem, every terminal-pair also comes with a
penalty $q_i$, and the algorithm can opt to not satisfy the request
by incurring this value in the objective. We give the extension
of our results to the corresponding prize-collecting problems 
in~\Cref{sec:prizes}.
}
\full{
\begin{theorem} \label{thm:prizes}
For each of the above problems, there is an online algorithm with matching running time and competitive ratio for the corresponding prize-collecting version.
\end{theorem}
}

In addition to the new results mentioned above, we can also use 
Theorem~\ref{thm:main} to give alternative proofs (with slightly 
worse polylog factors) of some recent results in online network design. 
By combining Theorem~\ref{thm:main} with the $\polylog(n)$-competitive
algorithm for online group Steiner Tree due to Alon~\etal \cite{AlonAABN06}, 
we obtain a $\polylog(n)$-competitive online algorithm for the 
group Steiner forest problem -- a result shown earlier by
Naor~\etal \cite{NaorPS11}.  Similarly, by combining
Theorem~\ref{thm:main} with the $\polylog(n)$-competitive online
algorithm for the node-weighted Steiner tree problem due to
Naor~\etal \cite{NaorPS11}, we obtain a
$\polylog(n)$-competitive online algorithm for the node-weighted 
Steiner forest problem -- a result shown earlier by Hajiaghayi~\etal \cite{HajiaghayiLP13}.

\submit{
 \paragraph{Height Reduction Theorem.} One of the technical tools that we use repeatedly in this paper is the following result, which builds on the work of Helvig~\etal~\cite{HelvigRZ01}. Due to lack of space, we defer its proof to the full version.}
\full{
 \smallskip \noindent {\bf Height Reduction Theorem.} One of the technical tools that we use repeatedly in this paper is the following result, which builds on the work of Helvig~\etal~\cite{HelvigRZ01}. We give the proof in~\Cref{app:layering}.}

\begin{theorem} \label{thm:layering}
	Given a directed graph $G = (V,E)$ with edge costs $c_e$ and
	lengths $\ell_e$, for all $h>0$, we can efficiently find an {
	upward directed, layered} graph $G^\upp_h$ on $(h+1)$
	levels and edges (with new costs and lengths) only between
	successive levels going from bottom (level $h$) to top (level $0$), such that each layer
	has $n$ vertices corresponding to the vertices of $G$, and, for any
	set of terminals $X$ and any root vertex $r$,
	\begin{OneLiners}
		\item[(i)] the optimal objective value of the {\em single-sink}
		buy-at-bulk problem to connect $X$ (at level $h$) with $r$ (at
		level $0$) on the graph $G^{\upp}_h$ is at most $O(h k^{1/h})
		\phi$, where $\phi$ is the objective value of an optimal solution
		of the same instance on the original graph $G$;
		\item[(ii)] given a integral (resp. fractional solution) of objective
		value $\phi$ for the single-sink buy-at-bulk problem to connect
		$X$ with $r$ on the graph $G^{\upp}_h$, we can efficiently
		recover an integral (resp. fractional solution) of objective value  at
		most $\phi$ for the problem on the original graph $G$.
	\end{OneLiners}
	Likewise, we can obtain a downward directed, layered graph
	$G^{\downn}_h$ on $(h+1)$-levels with edges going from top to
	bottom, with the same properties as above except for single-source
	instances instead.
\end{theorem}

\section{Proof of \Cref{thm:main} (Online Reduction to Single-Sink Instances)}
\label{sec:proofmain}

There are three main steps in the proof. In
Section~\ref{sec:layered}, we describe the composite LP which is a
relaxation of optimal junction-tree solutions (for technical reasons,
we first need to pre-process the graph). Next, in
Section~\ref{sec:solveLP}, we show how to fractionally solve the LP
online. Third, in Section~\ref{sec:roundLP}, we show how to {\em
partially round} the LP online. The resulting
solution then decomposes as fractional solutions to different
single-sink instances. Finally, we use the bounded integrality gap and the
online algorithm for SS-BB to wrap up the proof in
Section~\ref{sec:wrap}.

\subsection{The Composite-LP Relaxation: \ref{lp1:eq1}}
\label{sec:layered}

We first apply~\Cref{thm:layering} with $h = \Theta(\log n)$ to
obtain layered graphs $G^\upp$ (resp., $G^\downn$) of height $O(\log
n)$ where all the edges are directed upward (resp. downward);
see~\Cref{fig:gup} for an illustration.
The
reason for this preprocessing is that the length of the \siti paths
appear as a factor in our final competitive ratio and the above step
bounds it to a logarithmic factor. Recall that the graph $G^\upp$
(resp., $G^\downn$) approximately preserves the single-sink (resp.,
single-source) solutions for any set of terminals and any root.
After this step, we can
imagine that all the roots (of the single-sink instances we will
solve) are vertices in level $0$, and all the terminals will be
vertices in level $h = \Theta(\log n)$.
For clarity of presentation, we refer to
the root and terminal vertices by the same name in both $G^\upp$ and
$G^\downn$ (even though the graphs are completely disjoint).
Overloading notation, let $V$ denote the vertices in level $0$ in
both $G^\upp$ and $G^\downn$, and let $E$ be the union of the edge
sets of $G^\upp$ and $G^\downn$. Furthermore, the cost $c_e$ and
length $\ell_e$ of these edges are inherited from~\Cref{thm:layering}.

\begin{figure}[h]
	\centering
	\includegraphics[scale=0.6]{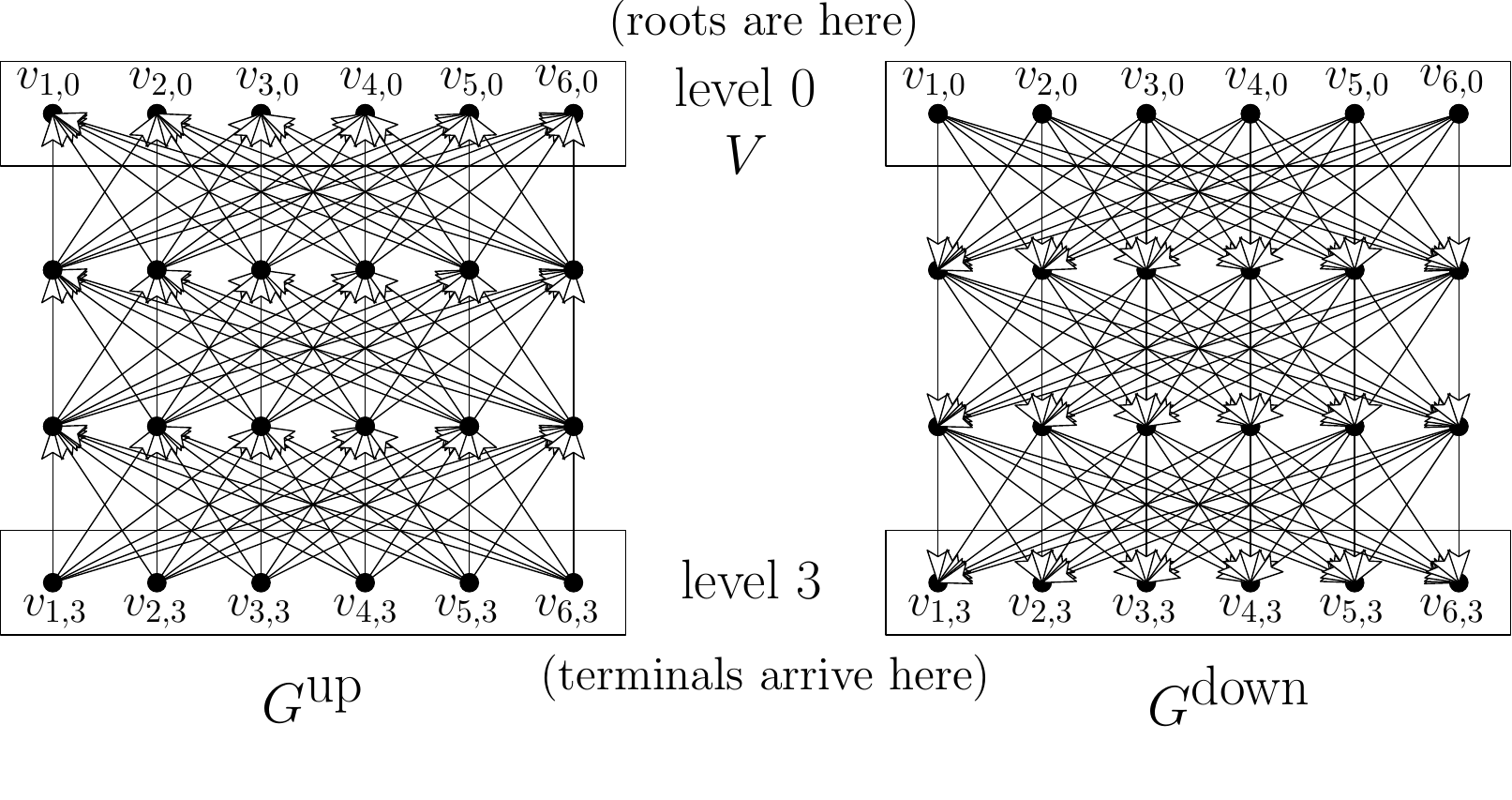}
	\caption{The graphs $G^\upp$ and $G^\downn$ with $h=3$, where the original graph $G$ has $6$ vertices $\{v_1, v_2, \ldots, v_6\}$.}
	\label{fig:gup}
\end{figure}

Now, using the junction-tree decomposition with approximation factor $\alpha$, we get the following lemma.

\begin{lemma} \label{lem:opt-jn}
	There exists a set $R^* \sse V$ of root vertices, a partition
	$\Pi^* := \{ \pi_{r} \, : \, r \in R^* \}$ of the terminal-pairs in
	$\cX$, a collection of in-trees $\{ T^\upp_{r} \, : \, r \in R^*\}$
	rooted at $r$ in $G^\upp$, and a collection of out-trees $\{
	T^\downn_{r} \, : \, r \in R^* \}$ rooted at $r$ in $G^\downn$ such
	that
	\begin{itemize}
		\item [(i)] Each $\siti \in \cX$ belongs to $\pi_r$ for some $r \in R^*$.
		\item [(ii)] For each $r \in R^*$, the in-tree $T^\upp_r$ is a
		feasible solution to the single-sink buy-at-bulk problem
		connecting $\{s_i \, : \, \siti  \in \pi_r\}$ to $r$ in $G^\upp$; likewise,
		the out-tree $T^\downn_r$ is a feasible solution to the
		single-source buy-at-bulk problem connecting $r$ to $\{t_i \, :
		\,  \siti  \in \pi_r\}$ in $G^\downn$.
		\item [(iii)] The sum of objective values of the single-source
		and single-sink solutions is at most $O(\alpha \log n) \cdot
		\opt(\cI)$.
	\end{itemize}
\end{lemma}
\full{
\begin{proof}
	Since the junction-tree approximation ratio of the given instance
	is $\alpha$, there exists a junction-tree solution given by a set
	of roots $R^*$ and a partition $\Pi^* := \{ \pi_{r} \, : \, r \in
	R^* \}$ such that the total objective value of all the single-sink
	and single-source junction-trees is at most $\alpha\opt(\cI)$.
	Moreover, by~\Cref{thm:layering}, because we choose the height $h =
	\Theta(\log n)$, the objective value of each single-sink and
	single-source solution in $G$ increases by a factor of $O(\log n)$
	in $G^\upp$ and $G^\downn$ respectively. Thus, the overall
	objective value of
	the resulting junction-trees in the graphs $G^\upp$ and $G^\downn$
	is $O(\alpha \log n) \cdot \opt(\cI)$.
\end{proof}
}

The above lemma motivates the LP relaxation given in
Fig.~\ref{fig:full-lp} which seeks to assign each \siti pair
to some rooted instance, and then minimizes the total fractional
objective value
of the rooted instances. Each individual rooted instance is
represented by an inner-LP (see the boxed constraints in
Fig.~\ref{fig:full-lp}).

\begin{figure}[!h]
\begin{subequations}
\begin{alignat}{2}
    \text{minimize }   & \sum_{r \in V} \sum_{e \in E} c_e x^r_{e} + \sum_{\siti \in \cX} \sum_{r \in R} \sum_{e \in E} \ell_e \left( f^r_{(e,s_i)} + f^r_{(e,t_i)} \right) \tag{MC-BaB LP} \label{lp1:eq1} \\
    \text{s.t } & \sum_{r \in V} z_{ir}  \geq 1 & & \forall i \label{lp1:eq2} \\
                       & z_{ir} \geq 0 \ &\ & \label{lp1:eq2a}
\end{alignat}
\begin{empheq}[box=\fbox]{alignat=2}
                       & \{ f^r_{(e, s_i)} \}  \text{ define a flow from }  s_i  \text{ to }  r  \text{ of value }  z_{ir} \text{ in }  G^\upp & \qquad  & \forall i, \forall r \in V \label{lp1:eq3a}   \\
                       & \{ f^r_{(e, t_i)} \}  \text{ define a flow from }  r  \text{ to }  t_i  \text{ of value }  z_{ir} \text{ in }  G^\downn &  & \forall i, \forall r \in V \label{lp1:eq3b}   \\
                       & f^r_{(e,s_i)} \leq x^r_e & &  \forall i, \forall e, \forall r \in V \label{lp1:eq6} \\
                       & f^r_{(e,t_i)} \leq x^r_e & &  \forall i, \forall e, \forall r \in V \label{lp1:eq7} \\
                       & x^r_e\geq 0,\ f^r_{(\cdot)} \geq 0\ &\ & \label{lp1:eq8}
\end{empheq}
\end{subequations}
\caption{Composite LP for MC-BB. \Cref{lp1:eq2,lp1:eq2a} form the outer-LP;
(\ref{lp1:eq3a})-(\ref{lp1:eq8}) form the inner-LPs.}
\label{fig:full-lp}
\end{figure}

In the LP, $z_{ir}$ denotes the extent to which the pair \siti chooses
root $r$ to route its flow. Within each inner-LP corresponding to a
root $r$, $\{x^r_e\}$ are the variables which denotes whether edge
$e$ is used to route flow in the corresponding rooted instance, and
$f^r_{(e, s_i)}$ (resp. $f^r_{(e, t_i)}$) denotes the amount of flow
$s_i$ sends (resp., $t_i$ receives) along $e$ to (resp., from) root
$r$. 
Observe that if the $z_{ir}$ variables are integral, then the
inner-LP corresponding to every root $r$ constitutes a feasible
fractional solution to \eqref{sslablp} for the single-sink instance
$\cI' = (G^\upp, \cX')$ where $\cX' = \{ (s_i, r) \! : \! z_{ir} = 1
\}$ and the single-source instance $\cI'' = (G^\downn,\cX'')$
where $\cX'' = \{ (r,t_i) \! : \! z_{ir} = 1\}$.
The next lemma, which bounds the optimal value of the \ref{lp1:eq1},
follows directly from Lemma~\ref{lem:opt-jn}.

\begin{lemma}
\label{lp:rel}
	The optimum value of~\eqref{lp1:eq1} is $O(\alpha \log n)\cdot
	\opt(\cI)$.
\end{lemma}
%

\subsection{An online fractional algorithm for the
\ref{lp1:eq1}}\label{sec:solveLP}

\begin{theorem} \label{thm:lponline}
	There is a randomized, polynomial-time online algorithm that
	returns a feasible fractional solution for \eqref{lp1:eq1} of value
	at most $O(\alpha \log^3 n)\cdot \opt(\cI)$.
\end{theorem}

\noindent
In the remainder of the subsection, we prove the above theorem. We
remark that the overall reduction uses~\Cref{thm:lponline} as a
black-box and the time-constrained reader can skip the proof and move
to Section~\ref{sec:roundLP}.

To simplify the exposition, we assume that we know the cost of an
optimal solution $\opt(\cI)$ up to a constant factor, using a standard doubling trick\footnote{Suppose our online algorithm has a competitive ratio of $\alpha$, and the true cost of an optimal solution is $c^*$. Then, we begin with an initial guess for the optimal cost, and run the online algorithm assuming this guess is the correct estimate for $c^*$. If our online algorithm fails to find a feasible solution of cost at most $\alpha$ times the current guess, we double our guess and run the online algorithm again. Eventually, our guess will exceed the optimal cost $c^*$ by at most a factor of two, and for this guess, the algorithm will compute a feasible solution of cost at most $2 \alpha c^*$. Moreover, since our guesses double every time, the total cost of the edges bought by the online algorithm over all the runs across different guesses is at most $\Theta(\alpha c^*)$.}.  Once we
know $\opt(\cI)$, by re-scaling all the parameters in the problem, we
may assume that it equals $1$. Next, we delete any edge
in $G^\upp$ or $G^\downn$ that has cost $c_e$ or length $\ell_e$
larger than $1$ as such edges cannot participate in any optimal
solution. Subsequently, we initialize all $x^r_e$ variables to $1/n^5$.
Likewise, we initialize all $z_{ir}$ variables to $1/n^5$ and also send
an initial flow of $1/n^5$ from each $s_i$ to $r$ in $G^\upp$ on an
arbitrary flow path from $s_i$ to $r$ and likewise from $r$ to $t_i$
in $G^{\downn}$. This setting ensures that the cost of the initial
solution is $o(1)$.

In the following, we partition the edge set $E$ into disjoint sets
$\{E_j: 0\leq j \leq h-1\}$, where $E_j$ denotes the set of edges
in $E$ between levels $j$ and $j+1$.  Furthermore, for clarity of
exposition, we describe a `continuous-time' version of the algorithm
where we increase the variables as a function of time. We note that
this algorithm can easily be discretized for a
polynomial-time\footnote{The polynomial is in the size of the input
to this algorithm, which for some of our algorithms/results is
quasi-polynomial in the size of the actual problem instance as
stated in the introduction.}
implementation. The algorithm is given as Algorithm~\ref{alg-fractional}.

\begin{algorithm}
	When a terminal-pair \siti arrives, we update the LP solution using
	the following steps:
	\begin{enumerate}[(1)]
		\item Let $R_i$ denote the set of roots $r$ in level $0$
		such that $s_i$ is connected to $r$ in $G^\upp$ and $r$ is
		connected to $t_i$ in $G^\downn$. For each $r \in R_i$,
		initialize a flow of value $1/n^5$ using any arbitrary flow path
		from $s_i$ to $r$ in $G^\upp$ and likewise from $r$ to $t_i$ in
		$G^\downn$. Also set $z_{ir} = 1/n^5$ for these roots.
		\item Repeat the following while $\sum_{r \in R_i} z_{ir} < 1$:
		\begin{enumerate}[(a)]
			\item Call an edge $e \in E$ {\bf tight} for root $r$ if $x^r_e =
			f^r_{(e,s_i)}$ or $x^r_e = f^r_{(e,t_i)}$.
			\item {\bf Edge Update:} For all {\em tight} edges $e \in E$,
			{\bf update} $x^r_e$ at the rate $\frac{d x^r_e}{d t} :=
			\frac{x^r_e}{c_e}$.
			\item {\bf Flow Update:} Solve the following min-cost max-flow
			problem for each $r \in R_i$: maximize $\Delta$ such that
			\begin{itemize}
				\item[-]  there exists a flow $\{ g^r_{(e,s_i)} \}$ sending
				$\Delta$ units of flow from $s_i$ to $r$ in $G^\upp$,
				\item[-] there exists a flow $\{ g^r_{(e,t_i)} \}$ sending
				$\Delta$ units of flow from $r$ to $t_i$ in $G^\downn$,
				\item[-]  {\em Capacity constraints}: $g^r_{(e,s_i)} \leq\medskip
				\noindent
				{\bf The Algorithm (Algorithm~\ref{alg-fractional}).}
				x^r_e/c_e$ and $g^r_{(e,t_i)} \leq x^r_e/c_e$ for all tight
				edges $e$,
				\item[-]  {\em Cost constraint:}  $\sum_{e} \ell_e \cdot
				g^r_{(e,s_i)} \leq z_{ir} $ and $\sum_{e} \ell_e \cdot
				g^r_{(e,t_i)} \leq z_{ir}$.
			\end{itemize}
			\item {\bf Update} $f^r_{(e,s_i)}$ at the rate $\frac{d
			f^r_{(e,s_i)}}{d t} := g^r_{(e,s_i)}$, and $f^r_{(e,t_i)}$ at
			the rate $\frac{d f^r_{(e,t_i)}}{d t} = g^r_{(e,t_i)}$ for all
			$e$, and update $z_{ir}$ at the rate $\frac{d z_{ir}}{d t} =
			\Delta$.
		\end{enumerate}
	\end{enumerate}
\caption{Online Fractional Algorithm for \eqref{lp1:eq1}}
\label{alg-fractional}
\end{algorithm}

We increase the $x$ variables on tight edges at a rate
inversely proportional to their cost, similar to the well-known
online set cover algorithm~\cite{AlonAABN09}.  However, the ``flow
constraints" are not pure packing (or covering) constraints and there
is no general-purpose way of handling them.  Indeed, we determine the rate of
increase of the flow variables by solving an
auxiliary min-cost max-flow subroutine which routes
incremental flows of equal value from $s_i$ to $r$ in $G^\upp$ and
from $r$ to $t_i$ in $G^\downn$ respecting capacity constraints (i.e., for
edges that are tight, the incremental flow is at most the rate of
increase of $x$).  This maintains feasibility in the inner LP.
Moreover, to bound the rate of increase in
objective, we enforce that the total length of the incremental flow
is at most $z_{ir}$ (this is the ``cost" constraint in the min-cost max-flow
problem). We stress that the incremental flows from  the auxiliary problem
dictate the {\em rate} at which we increase the original flow variables in
the LP. The final solution is feasible since the algorithm runs
until the outer-LP constraint is satisfied.

First, note that the total cost of initialization is $o(1)$ over all
the edge and flow  variables.
So it suffices to bound the cost of the updates.  The
next lemma relates the total cost of the updates to the total time
 $\tau$
for which the algorithm runs, and the subsequent lemma bounds $\tau$
in terms of $\opt(\cI)$.
%

\begin{lemma} \label{lem:dpdt}
	The LP objective value at the end of the above
	algorithm is $O(\log n)\cdot \tau$, where $\tau$ is the (continuous) time
	for which the algorithm runs.
\end{lemma}
\full{
\begin{proof}
\def\tight{\mathrm{tgt}}
	We show that at any time $t$, the rate of the increase of the
	LP objective value in the algorithm is at most $O(\log n)$; this proves the
	lemma.\footnote{We remark that the ``$\log n$'' corresponds to the
	number of levels in $G_h$ justifying the preprocessing step before
	the LP description.}  The objective increases because of increase
	in $x$ variables and flow variables $f$; we bound these separately.
	
	We first upper bound the objective increase due to the changes in
	the $x$ variables. Fix a level $j$ and let $E^{\tight}_{j}$ denote
	the set of tight edges in $E_{j}$ at time $t$. By definition,
	$\sum_{e\in E^{\tight}_j \cap G^\upp} x^r_e$ (resp., $\sum_{e\in
	E^{\tight}_j \cap G^\downn} x^r_e$) equals the total flow on these
	edges for the pair $(s_i, t_i)$.  Since the edges in $E_j$ form a
	cut separating $s_i$ from $t_i$, the total flow across this cut is
	at most $z_{ir}$.  Since $\sum_{r} z_{ir} < 1$, we have
	$\sum_r\sum_{e\in E^\tight_j} x^r_e < 2$.  Now, the rate of
	increase of each such tight edge is precisely $x^r_e/c_e$, which
	implies that the total rate of increase of the LP value due to
	the increase of $x$ is at most
	$$\sum_{r} \sum_{e \in E^\tight_{j}} c_e \frac{d x^r_e}{d t} \leq \sum_{r} \sum_{e \in E^\tight_{j}} x^r_e \leq 2.$$
	Summing over all levels gives the desired $O(\log n)$
	bound.
	
	Next, we upper bound the objective increase due to the changes in
	the $f$ variables. When these variables are updated, the total rate
	of increase of the objective due to the lengths of the $(s_i, r)$
	and $(r, t_i)$ flow paths is at most $z_{ir}$ --- this is precisely
	the ``cost" constraint in the auxiliary flow problem.  Hence the
	total rate of increase of flow lengths is at most $2$, completing
	the proof.
\end{proof}
}
\submit{
\begin{proofsketch}
	We show that at any time $t$, the rate of the increase of overall
	LP value by the algorithm is at most $O(\log n)$; this proves the
	lemma\footnote{We remark that the ``$\log n$'' corresponds to the
	number of levels in $G^\upp$ and $G^\downn$ justifying the
	preprocessing step.}.  The rate of
	increase of the LP objective value due to increases in $x$ is at most the total
	$x$-value on the tight edges of the current solution. For each $r$,
	each level is a cut, and so the total $x$-value on tight edges
	crossing a cut is at most the total flow crossing which is
	$z_{ir}$. Summing over all roots $r$, and then summing over all
	levels, we get an $O(\log n)$ increase in LP objective value due to
	$x$-variables.  The rate of increase of LP objective contribution
	of $f$-variables corresponding to any root $r$  is at most
	$z_{ir}$; this is maintained by the `cost' constraint in the
	algorithm. Since $\sum_r z_{ir} < 1$, we get that the contribution
	of this is $<1$.
\end{proofsketch}
}

\full{\medskip\noindent
Given the above lemma, we are left to relate $\tau$ to $\opt(\cI)$ in
order to complete the proof of Theorem~\ref{thm:lponline}.}

\begin{lemma}\label{lem:tau-vs-opt}
	The time duration $\tau$ of the above algorithm satisfies $\tau =
	O(\alpha \log^2 n)\cdot \opt(\cI)$.
\end{lemma}

We will need several new definitions and auxiliary lemmas in order to
prove \Cref{lem:tau-vs-opt}.  Recall from~\Cref{lem:opt-jn} that we
can assume that the solution that we are comparing against is the set
of junction-trees defined by $T^\upp_r$ and $T^\downn_r$ for $r \in
R^*$.  Also, recall that the terminal-pairs are partitioned by the
groups $\Pi^* = \{\pi_r: r \in R^*\}$.  For every $\siti  \in \cX$,
let $P^*_{s_i}$ denote the path from $s_i$ to the root $r$ in
$T^\upp_r$ such that $\siti \in \pi_r$.  Similarly, let $P^*_{t_i}$
denote the path from $r$ to $t_i$ in $T^\downn_r$.  Let $ \ell(P) = \sum_{e \in P} \ell_e$ for any path $P$. \Cref{lem:opt-jn}
asserts that
\begin{equation}\label{eq:007}
	\sum_{r \in R^*} \left(\sum_{e \in T^\upp_r \cup T^\downn_r} c_e +
	\sum_{\siti \in \cX} \left( \ell(P^*_{s_i})  + \ell(P^*_{t_i})
	\right)\right) = O(\alpha\log n)\cdot \opt(\cI)
\end{equation}
To bound $\tau$ against the optimal junction-tree solution, we use
two sets of \emph{charging clocks}:
\begin{itemize}
	\item We maintain an {\em edge clock} on every $(e,r)$ pair such
	that $e \in T^\upp_r$ or $e \in T^\downn_r$, i.e., if $e$ is used
	by the optimal junction-tree solution in the single-source (or
	single-sink) instance corresponding to $r$.  In particular, note
	that if an edge $e$ is in multiple junction-trees, then it has a
	separate clock for each such tree.
	\item We maintain a  {\em terminal clock} on every terminal-pair
	$\siti \in \cX$.
\end{itemize}
The crucial invariant that we maintain is the following: {\em at any
time instant $t$, {\em at least} one clock ``ticks," i.e., augments
its counter at unit rate.} The overall goal would then be to bound
the total time for which all the charging clocks can cumulatively
tick.

First, we describe the rule for the ticking of the clocks. Fix a time
$t$, and let the terminal-pair \siti be the pair that is active at
time $t$.  Let $r$ denote the root vertex which \siti has been
assigned to in the optimal junction-tree solution
from~\Cref{lem:opt-jn}, i.e., $\siti \in \pi_r$. Now, consider the
flow-paths $P^*_{s_i}$ in $G^\upp$ and $P^*_{t_i}$ in $G^\downn$. We
can have one of two situations:
\begin{itemize}
	\item[-] If any variable $x^r_e$ is tight for any edge $e \in
	P^*_{s_i}  \cup P^*_{t_i}$ at time $t$, then the edge clock on the
	pair $(e,r)$ ticks at time $t$. If there are multiple such edges,
	then all the corresponding clocks tick.
	\item[-] Otherwise, both paths are free of tight edges. In this
	case, the terminal clock for \siti ticks at time $t$.
\end{itemize}

\begin{lemma}
\label{lma:edge-clock}
	For any pair $(e,r)$ such that $e \in T^\upp_r \cup T^\downn_r$,
	its edge clock ticks for $O(c_e \log n)$ time.
\end{lemma}
\full{
\begin{proof}
	Notice that $x^r_e$ is initialized to $1/n^5$ for all roots $r$,
	and increases at the rate
	\begin{equation}
		\frac{d x^r_e}{ d t} = \frac{x^r_e}{c_e}
	\end{equation}
	\emph{at all times when the edge clock on $(e, r)$ ticks}. To see
	why, consider a time $t$ when the clock on $(e,r)$ ticks, and let
	$\siti$ denote the active terminal-pair at time $t$. It must be
	that
	(i) \siti has been assigned to root $r$ in $\Pi^*$, and
	(ii) either $x^{r}_e = f^{r}_{(e,s_i)}$ or $x^{r}_e = f^{r}_{(e,t_i)}$.
	But in this case, we increase such variables at rate $x^{r}_e/c_e$
	in our algorithm (Step (2a)).
	Therefore, we can infer that the value of $x^{r}_e$ would be  $1$
	after the edge clock on $e$ has ticked for time $O(c_e \log n)$.
	But clearly, $e$ cannot be a tight edge for any subsequent
	terminal-pair \siti once $x_e$ reaches 1; therefore, the edge clock
	on $(e,r)$ ticks for $O(c_e \log n)$ time overall.
\end{proof}
}
\submit{
\begin{proofsketch}
	Whenever the $(e,r)$ clock ticks, edge $e$ is tight,
	in which case our algorithm increases $x^r_e$ at rate $\frac{dx^r_e}{dt}  =
	x^r_e/c_e$. Since all variables are initialized at $\frac{1}{n^5}$ and cannot
	exceed $1$, the number of ticks is 
	$O(c_e\ln n)$.
\end{proofsketch}
}

\begin{lemma}
\label{lma:terminal-clock}
	For every terminal-pair $\siti$ connected by the optimal
	junction-tree solution through the root vertex $r$, the total time
	for which the terminal clock ticks is at most $O(\log n)\cdot \max
	(\ell(P^*_{s_i}), \ell(P^*_{t_i}))$.
\end{lemma}
\full{
\begin{proof}
	Recall that if the terminal clock for $\siti$ is ticking at time
	$t$, then it must mean that no edge is tight on either path
	$P^*_{s_i}$ or $P^*_{t_i}$.  In this case, we show that the
	variable $z_{ir}$ increases at a fast enough rate, where $r$ is the
	root \siti is assigned to in the optimal junction-tree, i.e.,
	$\siti \in \pi_r$.  We show this by exhibiting a feasible solution
	to the auxiliary LP  considered in Step (2b) of the algorithm for
	root $r$.  Indeed, send the flow from $s_i$ to $r$ along
	$P^*_{s_i}$, and likewise from $r$ to $t_i$ along $P^*_{t_i}$.
	Also set the value of $\Delta$ to be $z_{ir}/\max(\ell(P^*_{s_i}),
	\ell(P^*_{t_i}))$. Clearly, on the edges of these flow paths, we do
	not have any capacity constraints since no edge is tight. So, the
	only constraints are the cost constraints which are satisfied by
	the choice of $\Delta$. Hence, the rate of increase of $z_{ir}$ is at
	least
	\begin{equation}
		\frac{d z_{ir}}{ d t} \geq \frac{z_{ir}}{\max(\ell(P^*_{s_i}),
		\ell(P^*_{t_i}))}
	\end{equation}
	{\em at all times when the terminal clock on \siti ticks}. This
	proves the claim, for otherwise the variable $z_{ir}$ would have
	reached $1$, and the algorithm would have completed processing
	\siti.
\end{proof}
}
\submit{
\begin{proofsketch}
	Whenever a terminal clock ticks, there are no tight edge on
	$P^*_{s_i}$ or $P^*_{t_i}$.  This implies that the auxiliary flow
	found by our online fractional algorithm in the Flow-Update step is
	at least $\frac{z_{ir}}{\max(\ell(P^*_{s_i}),\ell(P^*_{t_i}))}$ (a
	candidate solution is to send flow along the paths $P^*_{s_i}$ and $P^*_{t_i}$) which in turn
	implies $dz_{ir}/dt$ is at least that large. Since $z_{ir}$ is
	initialized at $1/n^5$ and cannot exceed $1$, the number of ticks
	is bounded by $O(\log n)\max(\ell(P^*_{s_i}),\ell(P^*_{t_i}))$.
\end{proofsketch}
}

Since at least one clock ticks at all times, the total time clocked
is at least $\tau$, the duration of the algorithm.
~\Cref{lma:edge-clock} and~\Cref{lma:terminal-clock} imply that
\full{\[
	\tau \leq O(\log n)\sum_{r \in R^*} \left(\sum_{e \in
	T^\upp_r \cup T^\downn_r} c_e + \sum_{\siti \in \cX} \left(
	\ell(P^*_{s_i})  + \ell(P^*_{t_i}) \right)\right)
\]
}
\submit{
	$\tau \leq O(\log n)\sum_{r \in R^*} \left(\sum_{e \in
	T^\upp_r \cup T^\downn_r} c_e + \sum_{\siti \in \cX} \left(
	\ell(P^*_{s_i})  + \ell(P^*_{t_i}) \right)\right)$
}
which together with \eqref{eq:007} completes the proof
of~\Cref{lem:tau-vs-opt}. Theorem~\ref{thm:lponline} follows from
\Cref{lem:dpdt} and \Cref{lem:tau-vs-opt}.

\subsection{Partial Online LP Rounding}\label{sec:roundLP}

We partially round the fractional solution returned
by~\Cref{thm:lponline} to obtain integral values for only the
outer-LP variables $z_{ir}$, i.e., each \siti pair is integrally
assigned to a root. The inner-LP variables $x$ and $f$ continue
to be fractional but represent {\em unit} fractional flow
from $s_i$ to $r$ and $r$ to $t_i$ for the \siti pairs assigned to
$r$. The partial rounding algorithm is given as Algorithm~\ref{alg-partialround}.


\begin{algorithm}[!h]
\begin{enumerate}[(1)]
\item {\bf Initialization:} Each root chooses a threshold $\tau_r \in [1/2n, 1/(3\log n)]$ uniformly at random.
\item {\bf Partial Rounding:} At each time, maintain the scaled solution $\tilde{x}^r_e = \min\left(1, x^r_e/\tau_r\right)$,  $\tilde{f}^r_{(\cdot)} = \min\left(1, f^r_{(\cdot)}/\tau_r\right)$. Also set $\tilde{z}_{ir} = 1$ if  $z_{ir} \geq \tau_r$.
\end{enumerate}
\caption{Online Partial Rounding Algorithm}
\label{alg-partialround}
\end{algorithm}

\begin{theorem}\label{thm:partialround}
	The scaled solution $(\tilde{x}, \tilde{f})$ component-wise
	dominates a feasible solution to the outer-LP, and the expected
	objective value of the scaled solution $(\tilde{x}, \tilde{f})$ is
	at most $O(\alpha \log^5 n) \cdot \opt(\cI)$.  Moreover, for each
	\siti, there exists at least one root $r$ such that $\tilde{z}_{ir}
	\geq 1$ with probability at least $1 - 1/n^3$.
\end{theorem}
\full{
\begin{proof}
	Since each root $r$ chooses its threshold $\tau_r$ independently
	and uniformly at random from $[1/2n, 1/\log n]$, the probability
	that $\tilde{z}_{ir} = 1$ is at least $z_{ir} \log n$ (since
	$\tilde{z}_{ir} = 1$ if and only if $\tau_r \leq z_{ir}$). Since
	this is independent for different roots, a standard
	Chernoff-Hoeffding bound application (see, e.g.,~\cite{MotwaniR97})
	shows that each \siti pair has $\tilde{z}_{ir} = 1$ for some root
	$r$ with probability at least $1 - 1/n^3$.  Moreover, the expected
	value of any variable $x^r_e$ is given by
	\begin{equation*}
		\Expectation{\tilde{x}^r_e} \leq \int_{\tau_r=1/2n}^{\log n}
		\frac{x^r_e}{\tau_r} \log n d\tau_r \leq  O(\log^2 n) x^r_e.
	\end{equation*}
	A similar argument shows that the expected values of scaled flow
	variables are also bounded by $O(\log^2 n)$ times their values in
	the fractional solution. This shows that the expected objective
	value of the $(\tilde{x},\tilde{f})$ solution is at most $O(\log^2
	n)$ times the value of $(x,f)$; by~\Cref{thm:lponline}, the latter
	is at most $O(\alpha \log^3 n) \opt(\cI)$.  Combining these facts
	gives us the desired bound on the value of the scaled solution.

	It remains to show that the scaled solution dominates a feasible
	solution to the LP. To this end, fix some root $r$ and let $\cX_r$
	denote the set of \siti pairs for which $z_{ir}=1$. We need to show
	that installing capacities of $\{\tilde{f}^r_{(e,s_i)}\}$ on the
	edges can support unit flow from $s_i$ to $r$ in $G^\upp$ for all
	$\siti \in \cX_r$. Suppose for contradiction that there is is a cut
	$Q$ separating $s_i$ from $r$ of capacity strictly smaller than
	$1$. This implies that every edge $e \in Q$ must have
	$f^r_{(e,s_i)} \leq \tau_r$; otherwise, we would have an edge $e$
	with $\tilde{f}^r_{(e,s_i)} = 1$, which contradicts our assumption
	on the cut capacity. But then the value of the min-cut is precisely
	$\left(\sum_{e \in Q} f^r_{(e,s_i)} \right)/\tau_r$, which must be
	at least $1$ because of the following two observations:
	(i) we know that $\{f^r_{(e,s_i)}\}$ is a feasible flow from $s_i$
	to $r$ of value $z_{ir}$ and hence it must be that $\sum_{e \in Q}
	f^r_{(e,s_i)} \geq z_{ir}$, and
	(ii) since $\tilde{z}_{ir} = 1$, it must be that $z_{ir} \geq
	\tau_r$. This contradicts the assumption that the cut capacity is
	strictly smaller than $1$.
	A similar argument shows that the variables
	$\{\tilde{f}^r_{(e,t_i)} \}$ can support unit flow from $r$ to
	$t_i$ for every $\siti$ with $\tilde{z}_{ir} = 1$.
\end{proof}
}
\submit{
\begin{proofsketch}
	For the first part we need to consider only roots $r$ for which
	$z_{ir}$ is rounded to $1$ (i.e., $z_{ir} \geq \tau_r$); but for
	any such $r$, the corresponding $x,f$ variables have also been
	scaled by $1/\tau_r$ factor which is {\em more} than the scaling
	factor of $z_{ir}$.  To bound the expectation, note
		$\Expectation{\tilde{x}^r_e} \leq \int_{\tau_r=1/2n}^{\log n}
		\frac{x^r_e}{\tau_r} \log n d\tau_r \leq  O(\log^2 n) x^r_e.$
	The third part follows from a standard randomized rounding analysis
	for set cover.
\end{proofsketch}
}

\subsection{Wrapping up: Invoking the Single-Sink Online
Algorithm}\label{sec:wrap}
We are now ready to put all the pieces together and present our overall
online multicommodity buy-at-bulk algorithm as Algorithm~\ref{alg3}.
$\ssalg$ is the online algorithm for SS-BB alluded to in point (iii)
of the statement of~\Cref{thm:main}.


\begin{algorithm}[!h]
{\bf when} \siti arrives
\begin{enumerate}[(1)]
\item  {\bf update} the fractional solution of the composite LP using the algorithm (Algorithm~\ref{alg-fractional}, Section~\ref{sec:solveLP}).
\item  {\bf partially round} the solution using algorithm in (Fig.~\ref{alg-partialround}, Section~\ref{sec:roundLP}).
\item  {\bf if}$(\exists r: z_{ir} \geq 1)$:  send both $s_i$ and $t_i$  to the instance of \ssalg with root $r$.
\item {\bf else}: buy a trivial shortest path between $s_i$ and $t_i$ on the metric $(c+\ell)$ and route along this path
\end{enumerate}
\caption{Online Multicommodity Buy-at-Bulk Algorithm}
\label{alg3}
\end{algorithm}

Clearly the algorithm produces a feasible solution; so we now argue
about the expected objective value.  Fix an \siti pair.  Since the
probability that a terminal-pair is not assigned to a root is $\leq
1/n^3$ (by~\Cref{thm:partialround}), the expected total contribution
of such unassigned terminal-pairs is $\leq 1 = \opt(\cI)$.  For a
root $r$, let $\pi_r$ be the terminal-pairs assigned to $r$. We know
that $(\tilde{x},\tilde{f})$ restricted to $\pi_r$ dominates a
feasible solution in \eqref{sslablp}. Letting $LP_r$ denote the
contribution of this restriction to the overall LP value, we get $\sum_r
LP_r = O(\alpha \log^5 n)\cdot \opt(\cI)$. By the integrality gap
condition, we get that $\opt_r$, i.e. the {\em integral} optimum
objective value of the instance generated by $r$ and $\pi_r$, is at
most $\beta\cdot LP_r$. (Here we are using the fact
from~\Cref{thm:layering} that moving to the layered instance does not
increase the integrality gap.) The objective value of the solution
produced by $\ssalg$ is at most $\gamma \cdot \opt_r$, where $\gamma$
is the competitive ratio of $\ssalg$. Putting these observations
together, we conclude that the overall objective value of the
solution returned by the online algorithm is $O(\alpha\beta\gamma
\log^5 n) \cdot \opt(\cI)$.  This completes the proof
of~\Cref{thm:main}.

\section{Online Directed Buy-at-Bulk}
\label{sec:directed}

In this section, we prove~\Cref{thm:gen-bab}. A natural approach is
to use the reduction given by Theorem~\ref{thm:main}. To this end, we need to
establish the following: the existence of a junction-tree scheme with
a good approximation; a good upper bound on the integrality gap for
single-sink instances of the LP given in
Section~\ref{sec:preliminaries}; and an online algorithm for single-sink
instances with a good competitive ratio.

Extending the work of Chekuri et al.~\cite{ChekuriEGS11}, Antonakopoulos~\cite{Antonakopoulos10} shows the existence of a
junction-tree scheme with approximation $O(\sqrt{k})$.
Unfortunately, the integrality gap of the LP relaxation is not very
well understood even for Steiner tree instances; \cite{ZosinK02}
gives an $\Omega(\sqrt{k})$ lower bound\footnote{However, in these
instances, $n$ is exponentially large in $k$. So, they do not rule out a
$\polylog(n)$ upper bound.} on the integrality gap for the Steiner
tree problem and no suitable upper bound is known.
We overcome this difficulty as follows. Instead of working with
general graphs, we pre-process the instance and obtain a
tree-like graph for which we can show that the LP has a good
integrality gap.  Finally, we give the first non-trivial online
algorithm for the directed single-sink buy-at-bulk problem.
These results, together with our reduction (\Cref{thm:main}), 
imply the online algorithm for MC-D-BB.

We devote the rest of this section to the proof of
\Cref{thm:gen-bab}; to aid the reader, we restate the theorem below.

\begin{theorem}\label{thm:genbab}
	For any constant $\eps> 0$, there is a
	$O(k^{\frac{1}{2}+\eps}\polylog(n))$-competitive, polynomial time
	randomized online  algorithm for the general buy-at-bulk problem.
\end{theorem}

\medskip
\noindent
{\bf Pre-processing step.} We first give our reduction from
general instances of the directed buy-at-bulk problem to much more
structured instances; the reduction loses a factor of
$O(k^{\frac{1}{2} + \eps})$ in the approximation ratio.

Let $h = \ceil{1/\eps}$. Given an instance $\cI = (G,\cX)$ of the
directed buy-at-bulk problem, we map it to a tree-like instance $\cJ
= (H,\cX)$ as follows. We start by applying
Theorem~\ref{thm:layering} to $G$ to obtain the graphs $G^\upp$ and
$G^\downn$; recall that these graphs are layered $(h+1)$-level
graphs with $n$ vertices (corresponding to the vertices in $G$) in
each level, and the levels are numbered $0, 1, \ldots, h$ with $0$
being called the root level. The graph $G^\upp$ has edges directed
from higher numbered levels to lower numbered levels, and $G^\downn$
has edges in the opposite direction. To facilitate the construction
of the graph $H$, we now create $n$ trees from $G^\upp$ and $n$ trees
from $G^\downn$ as follows.

For every ``root vertex'' $r$ at level $0$ in $G^\upp$ (resp.
$G^\downn$), the tree $T^\upp_r$ (resp. $T^\downn_r$) is constructed
as follows:
\begin{itemize}
	\item The $0^{\text{th}}$ layer of $T^\upp_r$ has just one 
	vertex -- the root $r$.
	\item For each $i$ such that $1 \leq i \leq h$, the $i$-th
	layer of $T^\upp_r$ contains all $(i + 1)$-length tuples
	$(r,v_1,\ldots,v_i)$ where $v_j$ is a vertex present in the
	$j$-th layer of $G^\upp$.
	\item For every edge $e = (v_{i},v_{i-1}) \in G^\upp$, there is an
	arc from $(r,v_1,\ldots,v_{i-1},v_i)$ to $(r,v_1,\ldots,v_{i-1})$
	inheriting the same cost $c_e$ and length $\ell_e$.
	%
\end{itemize}

\medskip\noindent
Therefore each tree $T^\upp_r$ is an in-arborescence, with all edges
directed towards the root. The tree $T^\downn_r$ is constructed
analogously except all edges are directed away from the root. In the
following, we use the term \emph{leaves} to refer to the vertices on
layer $h$ of these trees.

After performing the above operation for every root vertex $r$ in
level $0$ of $G^\upp_r$ and $G^\downn_r$, we have $2n$ trees. Then
the final graph $H$ is obtained as follows (see
Figure~\ref{fig:hgraph}). For each root $r \in V$, we first add an
arc from the root of $T^\upp_r$ to $T^\downn_r$ of zero cost and
length. Finally, for every \siti pair, we add the vertices $s_i$ and
$t_i$ to $H$ and the following arcs connecting them to the trees: for
each tree $T^\upp_r$, we add an arc from $s_i$ to each leaf of
$T^\upp_r$ of the form $(r,v_1,\ldots, v_h)$ with $v_h = s_i$; for
each tree $T^\downn_r$, we add an arc to $t_i$ from each leaf of
$T^\upp_r$ of the form $(r,v_1,\ldots, v_h)$ with $v_h = t_i$. These
new arcs have zero cost and length (i.e., $c_e = \ell_e = 0$).

This completes the construction of $H$. Note that the graph $H$ has
$n^{O(h)}$ vertices and a similar number of edges. Our new instance
is $\cJ = (H, \cX)$ and we will apply \Cref{thm:main} to this
instance.

\begin{figure}[t]
	\centering
	\includegraphics[scale=0.6]{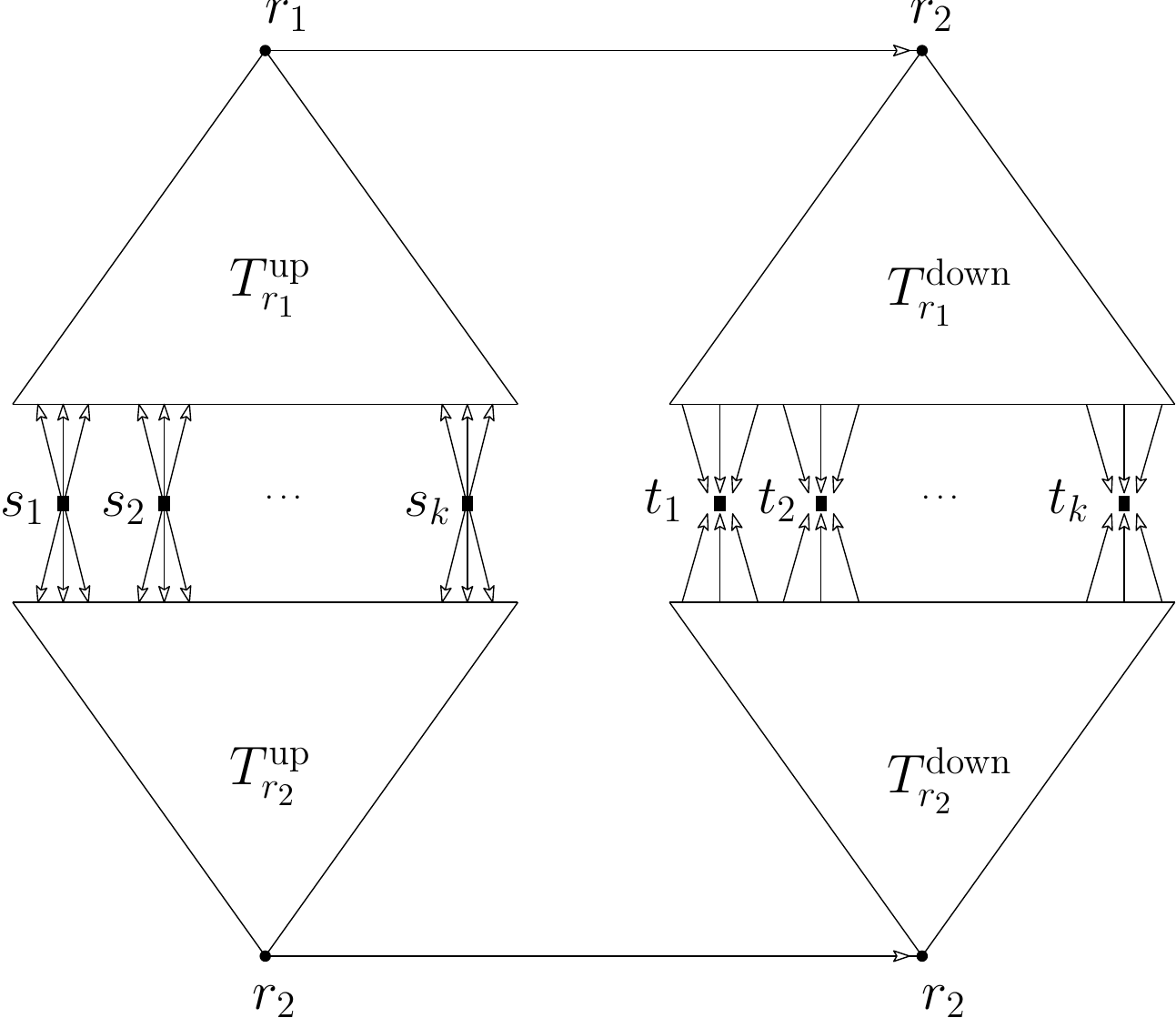}
	\caption{Construction of graph $H$.}
	\label{fig:hgraph}
\end{figure}

We first relate the objective values of $\cJ$ and $\cI$.

\begin{lemma}\label{clm:any-feasible-solution-in-J-is-a-junction-tree}
	Every feasible solution for $\cJ$ is a
	junction-tree solution.
\end{lemma}
\begin{proof}
	Note that any \siti path in $H$ has the following structure: $s_i$
	connects to a leaf node of $T^\upp_r$ for some $r \in V$, then
	continues to the root $r$, then traverse the edge to the root of
	$T^\downn_r$, then goes down to a leaf of $T^\downn_r$ and finally
	connects to $t_i$. Thus, for any feasible solution for $\cJ$, the
	\siti pairs can be partitioned based on the root $r$ through which
	they connect.
\end{proof}

\begin{lemma}\label{clm:JtoI}
	Any feasible solution for $\cJ$ can be mapped to a feasible
	solution --- in fact, a junction-tree solution --- for $\cI$ of equal
	or smaller objective value.
\end{lemma}
\begin{proof}
	Note that from the previous lemma, any feasible solution $S$ in
	$\cJ$ is a junction-tree solution. Therefore, there is a partition
	$\Pi_S$ of the \siti pairs depending on which root vertex they are
	using to connect. Moreover, it follows from our construction of the
	trees in $H$ that any edge in $T^\upp_r$ (resp. $T^\downn_r$)
	corresponds to an edge in $G^\upp$ (resp.  $G^\downn$). Therefore,
	if we map each edge appearing in solution $S$ to its
	corresponding edge in $G^\upp$ or $G^\downn$, we obtain a mapping
	from each junction tree of $S$ rooted at $r$ to a junction tree in
	$G^\upp \cup G^\downn$ rooted at $r$ that is connecting the same
	subset of pairs. Finally, by \Cref{thm:layering}, each junction
	tree in $G^\upp \cup G^\downn$ rooted at $r$ can be mapped, without
	increasing the objective value, to a junction tree in $G$ rooted at
	$r$ that is connecting the same subset of pairs. This completes the
	proof of the lemma.
\end{proof}


\begin{lemma}\label{clm:optJtooptI}
	$\opt(\cJ) \leq O(h k^{{1}/{h}}) \opt_\junc(\cI)$, where
	$\opt_\junc(\cI)$ is the objective value of an optimal
	junction-tree solution for $\cI$.
\end{lemma}
\begin{proof}
	Consider the optimal junction tree solution for $\cI$. Let the
	optimum partition be $\Pi = (\pi_{r_1},\ldots,\pi_{r_q})$ where
	$R^* = \{r_1, r_2, \ldots, r_q\}$ is the set of roots of the
	junction trees.  For each $r \in R^*$, let $X_r = \{ s_i \, : \,
	\siti \in \pi_r \}$ be the set of sources of $\pi_r$ and $Y_r$ be
	the corresponding sinks. From~\Cref{thm:layering}, we know that the
	optimum objective value of any one single-sink problem connecting
	$X_r$ to $r$ in $G^\upp$ is at most $O(h k^{1/h})$ times the
	objective value of the optimum solution connecting each source in
	$X_r$ to $r$. An analogous upper bound holds for every optimal
	single-source solution connecting $r$ to each sink in $Y_r$.
	Therefore, we get that the total sum of objective values of each of
	the junction trees in $G^\upp$ and $G^\downn$ is at most
	$O(hk^{1/h}) \opt_\junc(\cI)$. Now notice that any solution $S_G$
	for $(G^\upp,X_r)$ can easily be ``simulated" by a solution $S_T$ in
	the tree $T^\upp_r$: indeed, for every root-vertex path $(r, v_1,
	v_2, \ldots, v_i)$ in the solution $S_G$, include the edge from
	$(r, v_1, v_2, \ldots, v_i)$ to $(r, v_1, v_2, \ldots, v_{i-1})$ in
	$S_T$ (recall the vertices in $T^\upp_r$ exactly correspond to such
	root-vertex paths). It is easy to see that the objective value of
	the solution $S_T$ in $T^\upp_r$ is the same as that of $S_G$.
	Similarly, any solution for $(G^\downn, Y_r)$ can be simulated in
	$T^\downn_r$ with the same objective value. It follows that there
	is a feasible solution in $\cJ$ of objective value at most $O(h
	k^{1/h}) \opt_\junc(\cI)$.
\end{proof}

\begin{corollary}\label{cor:ItoJ}
	$\opt(\cJ) \leq O(k^{\frac12 + \eps})\opt(\cI)$.
\end{corollary}
\begin{proofsketch}
	Antonakopoulos \cite{Antonakopoulos10} shows that there exists a
	junction-tree solution of cost at most $O(\sqrt{k}) \opt(\cI)$. The
	corollary follows from this work and the fact that we set 
	$h = \Theta(1 / \eps)$.
\end{proofsketch}

\smallskip\noindent
Now we are ready to show that the new instance $\cJ$ has the
properties required by the reduction, i.e., Theorem~\ref{thm:main}
can be applied. In the
following lemma, a single-source (resp. single-sink)
\emph{sub-instance} $\cJ' = (H, \cT, v)$ of $\cJ = (H, \cX)$ is a
single-source (resp. single-sink) instance of the following form: the
graph is the same as in $\cJ$, namely $H$; the set of terminals $\cT$
is a subset of the sources (resp.  sinks) of $\cX$; the terminals
$\cT$ need to be connected to a root vertex $v \in V(H) \setminus \{
s_i, t_i \colon i \in [k]\}$.

\begin{lemma}\label{lem:done}
	Let $\cJ$ be the instance described above. Let $\alpha, \beta,
	\gamma$ be as in the statement of Theorem~\ref{thm:main}. We have
	\begin{itemize}
		\item [(i)] The junction-tree approximation factor of $\cJ$ is
		$1$ (i.e., $\alpha = 1$).
		\item [(ii)] The integrality gap of \eqref{sslablp} for any
		single-sink/source sub-instance $\cJ' = (H, \cT, v)$  is $O(h\log
		n\log k)$ (i.e., $\beta = O(h \log n \log k)$).
		\item [(iii)] There is a $O(h^2\log^2 n\log k)$-competitive
		algorithm for any single-sink/source sub-instance $\cJ' = (H,
		\cT, v)$ (i.e., $\gamma = O(h^2 \log^2 n \log k)$).
	\end{itemize}
\end{lemma}
\begin{proof}
	Property (i) follows
	from~\Cref{clm:any-feasible-solution-in-J-is-a-junction-tree}. Thus
	we focus on proving (ii) and (iii). In the following, we assume
	that we are working with a single-sink sub-instance $\cJ'$ of
	$\cJ$; the proof is very similar for single-source sub-instances
	and we omit it.

	In order to show (ii) and (iii), we will map the sub-instance $\cJ'
	= (H, \cT, v)$ to an instance of the {\em group Steiner tree}
	problem on a tree as follows.
	Since $v \in V(H) \setminus \{s_i, t_i \colon i \in [k] \}$, we
	have $v \in T^\upp_r \cup T^\downn_r$ for some $r$. We first
	consider the case when $v \in T^\upp_r$. In order to define the
	tree of the group Steiner tree instance, we start with the subtree
	$T_v$ of $T^\upp_r$ rooted at $v$. We add the following `dangling'
	edges to $T_v$ for each source $s_i \in \cT$: for each leaf vertex
	$u \in T^\upp_r$ such that $s_i$ has an edge in $T^\upp_r$ to $u$,
	we add a new vertex $u'$ to $T_v$ and connect it to $u$. Let $T$ be
	the resulting tree. We assign weights to the edges of $T$ as
	follows.  Each of the old edges $e \in T_v$ receives a weight equal
	to $c_e$.  Each of the new edges $uu' \in E(T) \setminus E(T_v)$
	receives a weight equal to $\ell(P_u)$,  where $P_u$ is the path of
	$T^\upp_r$ from $u$ to $v$. Finally, we define the following
	groups: for each source $s_i \in \cT$, we introduce a group $S_i$
	consisting of all the new vertices $u' \in V(T) \setminus V(T_v)$
	such that $s_i$ is connected to its partner $u$ in $T^\upp_r$ (that
	is, $T^\upp_r$ has an edge from $s_i$ to $u$). In the resulting
	group Steiner tree instance, the goal is to connect all of the
	groups $\cS = \{S_i \colon s_i \in \cT \}$ to the root $v$ using a
	minimum weight subtree of $T$; we let $(T, \cS, v)$ denote this
	instance. 
	
	Now the key claim is that the feasible solutions to the
	single-sink buy-at-bulk $(H, \cT, v)$ are in a one-to-one
	correspondence with feasible solution to the group Steiner tree
	instance $(T, \cS, v)$; moreover, the objective value of a
	solution to the former is equal to the weight of the solution to
	the latter. To see this, consider a feasible solution $\cP$ for the
	single-sink buy-at-bulk instance. Note that, for each source $s_i
	\in \cT$, $\cP$ has a path connecting $s_i$ to $v$; it follows from
	our construction of $H$ that this path consists of an edge from
	$s_i$ to a leaf $u$ of $T^\upp_r$ followed by the unique path in
	$T^\upp_r$ from $u$ to $v$. Thus we can construct a feasible group
	Steiner tree solution by connecting each group $S_i$ using the path
	of $T$ from $u'$ to $v$, where $u'$ is the partner of the leaf $u$
	of $T^\upp_r$ through which $s_i$ connects to the root in the
	buy-at-bulk solution. The weight of the edge $u'u$ captures the
	$\ell$-cost of $s_i$'s path and the weight of the path from $u$
	to $v$ captures the $c$-cost of $s_i$'s path. 

	Moreover, we can apply the same argument to fractional solutions to
	the two problems and show that there is a bijection between
	feasible fractional solutions to \eqref{sslablp} and feasible
	fractional solutions to the LP relaxation for group Steiner tree of
	Garg, Konjevod, and Ravi \cite{GKR}; as before, these corresponding
	solutions have the same objective values. Now the desired upper
	bound on the integrality gap of \eqref{sslablp} follows from the
	work of \cite{GKR} who showed that the integrality gap of the group
	Steiner tree LP is $O(\log N\log K)$ where $N = \max_i |S_i|$ is
	the maximum size of a group and $K$ is the number of groups. In our
	setting, $K \leq |\cX| \leq k$ and $N \leq n^h$. Therefore the
	integrality gap is $O(h \log{n} \log{k})$, which establishes
	property (ii).

	Moreover, notice that the above reduction can also be used to
	obtain an online algorithm for the single-sink (and single-source)
	sub-instances. Indeed, we simply use the online group Steiner tree
	algorithm of Alon \etal~\cite{AlonAABN04} which has a competitive ratio
	$O(\log^2 N\log K) = O(h^2 \log^2{n} \log{k})$. This proves
	property (iii).
\end{proof}

\medskip\noindent
Now we are ready to complete the proof of Theorem~\ref{thm:genbab}.
\vspace{-0.1in}
\begin{proofof}{Theorem \ref{thm:genbab}}
	Given the instance $\cI = (G, \cX)$, we construct the graph $H$ as
	described above; the time taken to do so is $n^{O(h)}$.  We 
	pass the instance $(H, \cX)$ to~\Cref{thm:main} and,
	using~\Cref{lem:done}, we obtain an online algorithm that, for any
	collection of pairs $\cX$ and any adversarial ordering of $\cX$,
	returns a solution of cost $O(\polylog(n)) \cdot \opt(\cJ)$. 
	By~\Cref{clm:optJtooptI} and~\Cref{cor:ItoJ}, we can map solutions
	for $(H, \cX)$ to solutions for $(G, \cX)$.
\end{proofof}

\medskip\noindent
We note that the approach described above also gives us new
online algorithms for the \emph{single-sink} buy-at-bulk problem on
directed graphs. For single-sink instances, the junction-tree
approximation is equal to $1$ and thus we save a factor of
$\sqrt{k}$. 

\begin{corollary} \label{cor:dir-st}
	There is a polynomial time $O(k^{\eps} \cdot \polylog(n))$-competitive
	online algorithm for the single-sink (or single-source) buy-at-bulk
	problem on directed graphs. The competitive ratio can be improved
	to $\polylog(n)$ if the running time can be quasi-polynomial in $n$.
\end{corollary}

\section{Online Single-sink, Undirected, Node-weighted Buy-at-Bulk}
\label{sec:node-weighted}
In this section, we prove~\Cref{thm:undirected-node-bab}. Again, we
use our main theorem (\Cref{thm:main}) and reduce the
multi-commodity buy-at-bulk problem to the single-sink version of the
problem. We combine the reduction theorem with the following results
from previous work. Chekuri {\em et al.}~\cite{ChekuriEGS11} show the
existence of a junction-tree scheme with approximation factor $O(\log
k)$. Moreover, the natural LP relaxation for the single-sink
buy-at-bulk problem on graphs with node costs is also $O(\log
k)$~\cite{ChekuriHKS10}. For the single-source online algorithm, we
resort to the algorithms from the previous section for the more
general directed single-sink buy-at-bulk problem
(see~\Cref{cor:dir-st}). We now obtain the desired result by the following
parameter settings:
\begin{eqnarray*}
	\alpha & = & O(\log k) \\
	\beta & = & O(\log k)\\ 
	\gamma & = & O(\polylog(n)) \\
	T & = & n^{O(\log n)}.
\end{eqnarray*}

\section{Online Prize-Collecting Buy-at-Bulk} \label{sec:prizes}
In the prize-collecting version of the buy-at-bulk problem, each terminal pair \siti also comes with a {\em penalty} $q_i$
and the algorithm may choose not to serve this request and incur the penalty in the total cost.
We show that our online reduction framework (\Cref{thm:main}) can be easily modified to handle prize-collecting versions as follows.

\begin{theorem}\label{thm:pcbab}
Let $\cI$ be a buy-at-bulk instance and suppose the three conditions of~\Cref{thm:main} hold.
Then there is an $O(\alpha\beta\gamma\cdot\polylog(n))$-competitive online algorithm for the online,
{\em prize-collecting} buy-at-bulk problem on $\cI$ with arbitrary penalties.
\end{theorem}
\begin{proof}
We closely follow the proof of~\Cref{thm:main}. The first difference is in the LP-formulation.
Now, for each \siti pair we have an extra variable $z_{i,0}$ which indicates whether we choose to discard this pair (and pay the corresponding penalty)
or not. We point out the differences with \eqref{lp1:eq1}. The new objective function is
\[
 \text{minimize }   \qquad \sum_{r \in V} \sum_{e \in E} c_e x^r_{e} + \sum_{\siti \in \cX} \sum_{r \in R} \sum_{e \in E} \ell_e \left( f^r_{(e,s_i)} + f^r_{(e,t_i)} \right) + \sum_{\siti\in\cX} q_i z_{i,0}
\]
and \eqref{lp1:eq2} is replaced by
\[
\sum_{r \in V} z_{ir}  + z_{i,0} \geq 1 \qquad \forall i
\]
Observe that the optimum value of this modified LP is at most $O(\alpha\log n)$ times the optimum: set $z_{i,0} = 1$ for the pairs the integral optimum solution does not connect,
and for the rest apply~\Cref{lp:rel}.
Also observe that the modified LP can be thought of the old LP on a modified instance where the graphs $G^\upp$ and $G^\downn$ (obtained from~\Cref{thm:layering}) have another vertex $``0"$ at the root level, and each $s_i$ has a direct path from $s_i$ to $``0"$ in $G^\upp$ with total length $q_i/2$ and no fixed cost,
and similarly, each $t_i$ has a path from $``0"$ to $t_i$ in $G^\downn$ with total length $q_i/2$ and no fixed cost. 
The rest of the proof now follows exactly as in~\Cref{sec:proofmain}, by also including the special vertex as a possible root while rounding to make the outer LP variables integral.
\end{proof} 

\section{Conclusion}
In this paper, we gave the first polylogarithmic-competitive online
algorithms for the non-uniform multicommodity buy-at-bulk problem.
Our result is a corollary of a generic online reduction
technique that we proposed in this paper for converting 
a multicommodity instance into several single-sink instances, 
which are often easier to design algorithms for.  We believe that this 
reduction will have other applications beyond the buy-at-bulk framework, 
and illustrate this by showing that recent results on online node-weighted 
Steiner forest and online generalized connectivity directly follow from 
our reduction theorem.
Our work also opens
up  new directions for future research. For instance, our algorithm for the
node-weighted problem runs in quasi-polynomial time, and a concrete
open question is to get a polynomial-time polylogarithmic-competitive algorithm for
the SS-N-BB problem (this suffices for MC-N-BB as well by our main
theorem). Another technical question concerns non-uniform demands. While our algorithm can be extended to the case of non-uniform demands, the approximation ratio incurs an additional $O(\log D)$ factor, where $D$ is the ratio of the largest to the smallest demand. It would be interesting to eliminate this dependence on $D$ since the corresponding offline results do not have this dependence. More generally, a broader question is to investigate other mixed packing-covering LPs that can be solved and rounded online.

\section*{Acknowledgements}
D. Panigrahi is supported in part by NSF Award CCF-1527084, a Google Faculty Research Award, and a Yahoo FREP Award.

\bibliographystyle{abbrv}
\bibliography{refs}

\appendix

\section{Reduction to Layered Instances (Proof of
Theorem~\ref{thm:layering})} \label{app:layering}
In this section, we prove
Theorem~\ref{thm:layering}, which is an extension of Zelikovsky's
`height reduction lemma' for the buy-at-bulk problem; Zelikovsky's
original Lemma was for a single metric, whereas in our setting there
is both a cost and a length metric.

We prove the up-ward case; the down-ward case follows analogously. In
order to simplify the notation, we remove the superscript $^\upp$.
For this reduction, we will adapt the notion of layered expansion of
a graph, which has been in the folklore for many years and has been
used recently by several papers (see, e.g.,\cite{ChekuriEGS11,
NaorPS11}). The $h$-level {\em layered expansion} of $G$ is a layered
DAG $G_h$ of $h+1$ levels (we index the level $0, 1, \ldots, h$)
defined as follows:

\begin{OneLiners}
	\item[(i)] For each $i$ such that $0 \leq i \leq h$, the vertices
	in level $i$ are copies of the vertices of $G$; we let $v_i$ to
	denote the copy of vertex $v \in V$ at level $i$.
	\item[(ii)] For each $i$ such that $1 \leq i \leq h$, there is a
	directed edge from every vertex in level $i$ to every vertex in
	level $i-1$. The fixed cost of an edge $(u_i, v_{i-1})$ is given by
	that of the shortest directed path $P^i_{uv}$ from $u_i$ to
	$v_{i-1}$ in $G$ according to the metric $c_e + k^{1 - i/h}
	\ell_e$. The length of this edge is set to be the length of the
	path $P^i_{uv}$ in the $\ell$ metric.
\end{OneLiners}

We now relate the optimal objective values for the two instances. One
of the directions of the reduction is straightforward.

\begin{lemma}
\label{lma:layered2directed}
	For any root $r$ and any set of terminals $X$, if there is a
	feasible integral/fractional solution of objective/LP value $\phi$
	for the single-sink buy-at-bulk problem connecting $X$ to $r$ on
	the $h$-level layered expansion $G_h$, then there is a feasible
	integral/fractional solution of objective/LP at most $\phi$ for the
	same problem in $G$.
\end{lemma}
\begin{proof}
	Note that for every edge in $G_h$, there is a corresponding path in
	$G$ with the property that the sum of edge costs and lengths on the
	path is at most the cost and length of the edge in $G$. Therefore,
	replacing the edges in the solution for the layered graph by the
	corresponding paths in $G$ yields a feasible solution in $G$
	without increasing the overall cost and length. Notice that the
	same ``embedding'' of edges in $G_h$ to paths in $G$ can be applied
	to the  fractional solution on $G_h$ as well. This shows property
	(ii) of the theorem statement.
\end{proof}

\medskip\noindent
The more interesting direction is to show that the optimal objective
value  on the layered graph $G_h$ can be bounded in terms of the
optimal objective value on the original graph $G$. To show this, we
will re-purpose the so-called ``height reduction'' lemma of Helvig,
Robins, and Zelikovsky~\cite{HelvigRZ01}. We restate the lemma in a
form that will be useful for us.


\begin{lemma}
	\label{lma:height-reduction}
	For any in-tree $T$ defined on the edges of $G$ that is rooted at
	$r$ and contains all the terminals in $X$, and for any integer $h
	\geq 1$, there is an in-tree $T'$ (on the same vertices as $G$ but
	over a different edge set) that is also rooted at $r$ and contains
	all the terminals, and has the following properties:
 	\begin{OneLiners}
 		\item[(i)] $T'$ contains $h+1$ levels of vertices, i.e., has height $h$.
 		\item[(ii)] $T'$ is an $k^{1/h}$-ary tree, i.e., each non-leaf vertex has
 		$k^{1/h}$ children.
 		\item[(iii)] Each edge $e' = (u',v')$ in tree $T'$ corresponds to the unique directed path $p_{e'}$
 		in $T$ from $u'$ to $v'$. Moreover, the number of terminals in the subtree
 		of $T$ rooted at $u'$ is exactly $k^{1-i/h}$, where $e'$ is an edge between
 		levels $i-1$ and $i$ of $T$.
 		\item[(iv)] Each edge in $T$ is in at most $2h k^{1/h}$ such paths $p_{e'}$
 		for edges $e'\in T'$.
 	\end{OneLiners}
\end{lemma}

For an edge $e' \in T'$, suppose we define its cost to be the
cost of the path $p_{e'}$, and its length to be the length of the
path $p_{e'}$. Then it is easy to see that the overall cost
$\phi_{T'}$ of tree $T'$ is $O(h k^{1/h})$ times that of tree $T$;
this is due to the following implications of the above lemma: (a) the
total (buying) cost of all edges in $T'$ is at most $2h k^{1/h}$
times that of $T$ since each edge is reused at most $2h k^{1/h}$
times, and (b) for any terminal $x \in X$, the edges on its path to
the root in $T'$ correspond to disjoint sub-paths in the unique path
between $x$ and $r$ in $T$, and hence the total length cost in $T'$
is at most that in $T$.

Using this lemma, we can now complete the reduction by
``embedding'' the tree $T'$ in the layered graph $G_h$.

\begin{lemma}
\label{lma:directed2layered}
	If there is a feasible solution of overall cost $\phi$ for the
	single-sink buy-at-bulk problem on $G$, then there is a feasible
	solution of overall cost $O(h k^{1/h} \phi)$ for the same problem
	on the $h$-level layered extension $G_h$.
\end{lemma}
\begin{proof}
	Let $T$ be the union of the paths in the optimum solution on the
	graph $G$. It's easy to see that $T$ is a directed in-tree.  First,
	we use Lemma~\ref{lma:height-reduction} to transform $T$ to tree
	$T'$ of height $h$. As noted earlier, the overall cost $\phi_{T'}$
	of $T'$ is $O(h k^{1/h} \phi)$. Now, we construct a feasible tree
	$T_h$ in $G_h$ using this solution $T'$ as follows: consider each
	edge $(u,v)$ in $T'$ where $u$ is at level $i$ and $v$ is at level
	$(i-1)$. Then, include the edge $(u_i, v_{i-1})$ in $T_h$.
	Clearly, since $T'$ connects all the terminals to the root, so does
	$T_h$. Moreover, notice that there is a 1-to-1 mapping between
	edges in $T'$ and edges in $T_h$.

	To bound the objective value of the subtree, we relate the objective
	value for each edge of the subtree in $G_h$ to its corresponding
	mapped edge in $T'$. First, note that the overall contribution of
	an edge $e' = (u',v')$ between layers $i$ and $i+1$ towards
	$\phi_{T'}$ is equal to the sum of costs and $k^{1-i/h}$ times the
	lengths of the edges on the associated path $p_{e'}$ from $u'$ to
	$v'$. This is because, by property (iii)
	of~\Cref{lma:height-reduction}, the number of demands in the
	subtree rooted at $u'$ is exactly $k^{1 - i/h}$ and all of them
	traverse this edge to reach $r$.  Next, we note that, by
	definition, the cost of the edge $(u'_i, v'_{i+1})$ between layers $i$
	and $i+1$ in $T_h$ is equal to the \emph{shortest directed path}
	from $u'$ to $v'$ in $G$ according to the metric $c_e + k^{1 - i/h}
	\ell_e$. Since we chose the shortest path, we get that the buying
	cost of edge $(u'_i, v'_{i+1}) \in T_h$ is at most the contribution
	of $(u',v')$ towards $\phi_{T'}$. Moreover, the total length cost
	is at most $k^{1-i/h}$ times the length of the shortest path, which
	is at most the fixed cost of $(u'_i, v'_{i+1}) \in T_h$ (again,
	this uses the fact that there are exactly $k^{1-i/h}$ terminals
	which route through this edge in $T_h$ also).  It therefore follows
	that the overall cost of the solution that we have inductively
	constructed in $G_h$ is at most twice the overall cost of $T'$,
	which is at most $O(h k^{1/h}) \phi$.
\end{proof}

\medskip\noindent
\Cref{thm:layering} follows from
\Cref{lma:layered2directed,lma:directed2layered}.


\end{document}